\documentclass[journal]{IEEEtran}

\usepackage{amsmath,amsthm,epsfig,amssymb,verbatim,amsopn,cite,subfigure,multirow}
\usepackage{balance,nopageno}
\usepackage{mathtools}
\usepackage[usenames,dvipsnames]{color}
\usepackage[all]{xy}  
\usepackage{url}
\usepackage{amsfonts}
\usepackage{amssymb}
\usepackage{epsfig}
\usepackage{epstopdf}
\usepackage{slashbox}
\usepackage{bm}

\newtheorem{Lemma}{Lemma}

\newtheorem{Corollary}[Lemma]{Corollary}

\newtheorem{Remark}{Remark}

  {\proof}{\proofend}
\newtheorem{proposition}{Proposition}

\newcommand{\ve}[1]{\boldsymbol{#1}}
\newcommand{\E}[1]{\mathbb{E}\left\{#1\right\}}

\newcommand{\vA}{\ve{A}}

\newcommand{\vH}{\ve{H}} \newcommand{\vh}{\ve{h}}
\newcommand{\vI}{\ve{I}}

\newcommand{\vU}{\ve{U}} \newcommand{\vu}{\ve{u}}
 \newcommand{\vv}{\ve{v}}
\newcommand{\vW}{\ve{W}} \newcommand{\vw}{\ve{w}}

 \newcommand{\vz}{\ve{z}}

\newcommand{\Ei}{\mathrm{E}_1}

\DeclareMathOperator*{\res}{Res}

\newcommand{\Us}{\mathsf{u}}

\newcommand{\Ds}{\mathsf{d}}
\newcommand{\AP}{\mathsf{a}}

\newcommand{\SINRd}{\mathsf{SINR_d}}
\newcommand{\SINRAP}{\mathsf{SINR_a}}

\newcommand{\SINRi}{\mathsf{SINR}_i}

\newcommand{\SNRd}{\mathsf{SNR_d}}
\newcommand{\SNRa}{\mathsf{SNR_a}}
\newcommand{\diag}{\mathsf{diag}}

\newcommand{\MRC}{\mathsf{MRC}}
\newcommand{\ZF}{\mathsf{ZF}}
\newcommand{\MRT}{\mathsf{MRT}}
\newcommand{\RFD}{R_{\mathsf{sum}}^{\mathsf{FD}}}
\newcommand{\RHDs}{R_\mathsf{{sum }}^{\mathsf{{HD-RC}}}}
\newcommand{\RHDd}{R_\mathsf{{sum }}^\mathsf{{HD-AC}}}
\newcommand{\RFL}{\tilde{R}_\mathsf{{sum }}^\mathsf{{FD}}}

\newcommand{\RFMM}{{R}_\mathsf{{sum }}^\mathsf{{FD-(MRC/MRT)}}}
\newcommand{\RFMZ}{{R}_\mathsf{{sum }}^\mathsf{{FD-(MRC/ZF)}}}

\newcommand{\MRCMRT}{\text{MRC}_{\mathsf{(rx)}}/\text{MRT}_{\mathsf{(tx)}}}
\newcommand{\MRCZF}{\text{MRC}_{\mathsf{(rx)}}/\text{ZF}_{\mathsf{(tx)}}}
\newcommand{\ZFMRT}{\text{ZF}_{\mathsf{(rx)}}/\text{MRT}_{\mathsf{(tx)}}}

\newcommand{\snr}{\mathsf{snr}}
\newcommand{\HD}{\mathsf{HD}}
\newcommand{\FD}{\mathsf{FD}}
\newcommand{\RC}{\mathsf{RC}}
\newcommand{\AC}{\mathsf{AC}}
\newcommand{\Bta}{\mathsf{Beta}}
\newcommand{\Trace}{\mathsf{trace}}

\newcommand{\PaHR}{ P_{\AP}^{\tiny{\mathsf{HD-RC}}}}
\newcommand{\PuHR}{ P_{\Us}^{\tiny{\mathsf{HD-RC}}}}
\newcommand{\PaHA}{ P_{\AP}^{\tiny{\mathsf{HD-AC}}}}
\newcommand{\PuHA}{ P_{\Us}^{\tiny{\mathsf{HD-AC}}}}

\newcommand{\EI}[1]{\mathbb{E}_{I_{\Ds,\Us}}\left\{#1\right\}}

\newcommand{\Erth}[1]{\mathbb{E}_{r,\theta}\left\{#1\right\}}
\newcommand{\Er}[1]{\mathbb{E}_{r}\left\{#1\right\}}

\newcommand{\Sn}{\sigma_n^2}
\newcommand{\Sap}{\sigma_{\AP\AP}^2}

\newcommand{\WT}{\vw_t}
\newcommand{\WR}{\vw_r}

\newcommand{\nr}{n_{\Us}}
\newcommand{\nt}{n_{\Ds}}

\newcommand{\Prob}{\textnormal{Pr}}

\usepackage{multibib}
\usepackage[nodisplayskipstretch]{setspace}
\newcites{Prim}{Very important papers}

\definecolor{light-gray}{gray}{0.65}

\newcounter{mytempeqcounter}

\title{Full-Duplex Radio for Uplink/Downlink Wireless Access with Spatially Random Nodes}

\author{\normalsize {Mohammadali Mohammadi,~\IEEEmembership{Member,~IEEE,}
 Himal A. Suraweera,~\IEEEmembership{Senior Member,~IEEE,}\\
 Yun Cao,~\IEEEmembership{}
 Ioannis Krikidis,~\IEEEmembership{Senior Member,~IEEE,}
  and Chintha Tellambura,~\IEEEmembership{Fellow,~IEEE}}%

  \thanks{%
Part of this work was supported by the Research Promotion Foundation, Cyprus,
under the Project FUPLEX with Pr. No. CY-IL/0114/02.
}

  \thanks{%
    M. Mohammadi is with the Faculty of Engineering, Shahrekord University, Shahrekord 115, Iran
    (e-mail: {m.a.mohammadi@eng.sku.ac.ir}).}
  \thanks{%
    Himal A. Suraweera is with the Department of Electrical and Electronic
Engineering, University of Peradeniya, Peradeniya 20400, Sri Lanka  (email:
    {himal@ee.pdn.ac.lk}). }%
    \thanks{I. Krikidis is with the Department of Electrical and Computer
Engineering, University of Cyprus, Nicosia 1678, Cyprus (email:
    {krikidis@ucy.ac.cy).}
    }
     \thanks{ C. Tellambura and Y Cao are with the Department of Electrical and Computer Engineering,
University of Alberta, Edmonton, AB T6G 2V4 Canada  (email:
    {chintha@ece.ualberta.ca; cao7@ualberta.ca).}
    }

      \thanks{%
Part of this work is accepted for presentation at the IEEE International Conference on Communications (ICC 2015), London, UK, June
2015.
}
}

\begin{document}

\maketitle
\thispagestyle{empty}
\vspace{-1em}

\begin{abstract}
A full-duplex (FD) multiple antenna access point (AP)  communicating with single antenna half-duplex (HD) spatially random users to support simultaneous uplink (UL)/downlink (DL) transmissions is investigated. Since FD nodes are inherently constrained by the loopback interference (LI), we study precoding schemes for the AP  based on maximum ratio combining (MRC)/maximal ratio transmission (MRT), zero-forcing and the optimal scheme for UL and DL sum rate maximization using tools from stochastic geometry. In order to shed insights into the system’s performance, simple expressions for single antenna/perfect LI cancellation/negligible internode interference cases are also presented. We show that FD precoding at AP improves  the UL/DL sum rate and hence a doubling of the  performance of  the HD mode is achievable. In particular, our results show that these impressive performance gains remain substantially intact  even if the LI cancellation is imperfect.  Furthermore, relative performance gap between FD and HD modes increases as the number of transmit/receive antennas becomes large, while with the MRC/MRT scheme, increasing the receive antenna number at FD AP, is more beneficial in terms of sum rate than increasing the transmit antenna number.
\end{abstract}

\begin{IEEEkeywords}
  Full-duplex, stochastic geometry, average sum rate, precoding, interference, performance analysis.
\end{IEEEkeywords}

\section{Introduction}

Due to the proliferation of devices such as smart phones, tablets, and personal digital assistants (PDAs) and the exponential growth of the number of subscribers, the world has witnessed a dramatic increase in wireless traffic recently.  The low-hanging fruit in terms of  spectral efficiency gains  of traditional point-to-point links has reached the  theoretical limits.  Only incremental gains of spectral efficiency appears feasible at this point. Thus, is it possible to significantly  improve overall spectral efficiency of networks any further?  Note  that wireless nodes operate  half-duplex in (HD) mode by separating  the uplink and downlink channels into orthogonal signaling (time or frequency) slots. FD mode (i.e., both uplink and downlink on the same channel at the same time), if possible, has the potential to double the spectral efficiency instantly.  The tremendous implications of FD wireless nodes will thus be not only  to transform for cellular  network designs radically and but also to double capacity, speed or number of subscribers of cellular networks ~\cite{DBLP:journals/jsac/SabharwalSGBRW14a}.

However, a  key challenge in implementing a FD transceiver is the presence of loopback interference (LI)~\cite{Riihonen:WCOM:2009,Riihonen:WCOm:2011,Riihonen:JSP:2011,Duarte:PhD:dis,Bliss:Asilomar:2012}. Since the LI is  caused by the self-transmitted signal in  the transceiver,  up until recently FD radio was considered practically infeasible.  This long-held  pessimistic view has been challenged in the wake of recent advances in antenna design and introduction of analog/digital signal processing solutions. To this end, several single and multiple antenna FD implementations have been developed through new LI cancellation techniques~\cite{Duarte:PhD:dis,Riihonen:JSP:2011,Bliss:Asilomar:2012,Khojastepour:Mobicom:2012,Sachin:NSDI:2014,Everett:JWCOM:2014, Korpi:JSAC:2014}. Antenna separation/radio frequency (RF) shielding, analog/digital and hybrid analog-digital circuit domain approaches can achieve significant levels of LI cancellation in single antenna FD systems. Multiple antenna LI suppression/cancellation techniques are largely based on the use of directional antennas and spatial domain cancellation algorithms. The implementation of single antenna FD technology with LI cancellation was demonstrated in~\cite{Duarte:PhD:dis}. The authors in~\cite{Riihonen:JSP:2011} and~\cite{Bliss:Asilomar:2012} characterized the spatial suppression of LI at FD relaying system. A multiple-input multiple-output (MIMO) FD implementation (MIDU) was presented in~\cite{Khojastepour:Mobicom:2012}, while~\cite{Sachin:NSDI:2014} reported design and implementation of an in-band WiFi-PHY based FD MIMO system.

FD systems find useful in several new applications that exploit their ability to transmit and receive simultaneously. Some of these examples include one-way~\cite{Himal:WCOM:FD:2014, Leonardo:JSAC:2014,Ngo:JSAC:2014} and two-way FD relay transmission~\cite{ZhenJSP:2015}, simultaneous sensing/operation in cognitive radio systems~\cite{Riihonen:CROWNCOM:2014} and reception/jamming for enhanced physical layer security~\cite{Zheng:JSPL2013}. Another possible advantageous use of FD communications is the simultaneous uplink (UL)/downlink (DL) transmission in wireless systems such as WiFi and cellular networks~\cite{Sanjay:CISS:2013,Girnyk:2013,Sachin:NSDI:2014,Sundaresan:Mobicom:2014, Panwar:ICC14}. However, such transmissions introduce LI and internode interference in the network as DL transmission will be affected by the LI and the UL user will interfere with the DL reception. Therefore, in the presence of LI and internode interference, it is not clear whether FD applied to UL/DL user settings can harness performance gains. To this end, several works in the literature have presented useful results considering topological randomness which is a main characteristic of wireless networks.

A new modeling approach that captures topological randomness in the network geometry and is capable of producing tractable analytical results is based on stochastic geometry. In~\cite{Sanjay:CISS:2013} a FD cellular analytical model based on stochastic geometry was used to derive the sum capacity of the system. However,~\cite{Sanjay:CISS:2013} assumed perfect LI cancellation and therefore, the effect of LI is not included in the results. The application of FD radios for a single small cell scenario was considered in~\cite{Panwar:ICC14}. Specifically in this work, the conditions where FD operation provides a throughput gain compared to HD and the corresponding throughput results using simulations were presented. In~\cite{DBeiYin:ACSSC}, the combination of FD and massive MIMO was considered for simultaneous UL/DL cellular communication. The information theoretic study presented in~\cite{Achaleshwar:acssc:DS13}, has investigated the rate gain achievable in a FD UL/DL network with internode interference management techniques.

In~\cite{Full:Nguyen:JSP:2013}, joint precoder designs to optimize the spectral and energy efficiency of a FD multiuser MIMO system were presented. However~\cite{DBeiYin:ACSSC,Achaleshwar:acssc:DS13,Full:Nguyen:JSP:2013} considered fixed user settings for performance analysis and as such the effect of interference due to distance, particularly relevant for wireless networks with spatial randomness, is ignored. Tools from stochastic geometry has been used to analyze the throughput of FD networks in~\cite{JeminLee:2014,Haenggi:arXiv,Venkateswaran:Infocom:2015}. Specifically,~\cite{JeminLee:2014} studied the throughput of multi-tier heterogeneous networks with a mixture of access points (APs) operating either in FD or HD mode. The throughput gains of a wireless network of nodes with both FD and HD capabilities has been quantified in~\cite{Haenggi:arXiv}, while~\cite{Venkateswaran:Infocom:2015} analyzed the mean network throughput gains due to FD transmissions in multi-cell wireless networks.

In this paper, we consider a wireless network scenario in which a FD AP is communicating with the single antenna spatially random user terminals to support simultaneous UL and DL transmissions. Specifically we consider a Poisson point process (PPP) for the DL users and assume that the scheduled UL user is located $d$ distance apart. The AP employs multiple antennas and therefore, precoding can be applied for proper weighting of the transmitted and received signals and spatial LI mitigation/cancellation. We develop a performance analysis and characterize the network performance using UL and DL average sum rate as the metric. Further, we present insightful expressions to show the effect of network parameters such as the user spatial density, the LI and internode interference (through UL/DL user distance parameter) on the average sum rate.

Our contributions are summarized as follows:
\begin{itemize}
\item

We consider both LI and internode interference and derive expressions for the UL and DL average sum rate when several precoding techniques are applied at the AP. Specifically, precoding schemes based on the maximum ratio combining (MRC)/ maximal ratio transmission (MRT), zero-forcing (ZF) for LI cancellation and the optimal precoding scheme for sum rate maximization are investigated. In order to highlight the system behavior and shed insights into the performance, simple expressions for certain special cases are also presented. Further, as an immediate byproduct, the derived cumulative density functions (cdfs) of the signal-to-interference noise ratios (SINRs) can be used to evaluate the system's UL and DL outage probability.

\item
Our findings reveal that for a fixed LI power, when the internode interference is increased, the $\MRCZF$ scheme achieves a better  performance than the $\MRCMRT$ scheme. On the other hand, by keeping the amount of internode interference constant, while decreasing the LI the $\MRCMRT$ scheme performs better than the $\MRCZF$ scheme. Moreover, in the presence of LI, increasing  the receive antenna number at the FD AP with the $\MRCMRT$ scheme, is more beneficial in terms of the average sum rate than increasing the number of transmit antennas at the AP.

\item We compare the sum rate performance of the system for
FD and HD modes of operation at the AP to elucidate the SNR regions where FD outperforms HD. Our results reveal that, the choice of the linear processing play a critical role in determining the FD gains. Specifically, optimal design can achieve up to $47\%$ average sum rate gain in comparison with HD scheme in all LI regimes. However, at high LI strength as well as high transmit power regime (i.e., $>30$ dB), FD mode with $\MRCMRT$ scheme becomes inferior as compared to the HD mode. Moreover, our results indicate that different power levels at the AP and UL user have a significant adverse effect to decrease the average sum rate in the HD mode of operation than the FD counterpart\footnote{In a multi-cell case, the exact level of performance gap between FD/HD operations will be further determined by the amount of co-channel interference generated in the network.  However, inter-cell interference can be reduced significantly in emerging networks by exploiting techniques such as interference coordination~\cite{Haenggi:TWC:ICIC:2014}.}.
\end{itemize}

The rest of the paper is organized as follows. Section~\ref{sec:system model and assumption} presents the system model and Section~\ref{sec:Performance Analysis} analyzes the UL and DL average sum rate of different precoding schemes applied at the AP. We compare the sum rate of the counterpart HD mode of operation as well as some special cases in Section IV. We present numerical results and discussions in Section~\ref{sec:Numerical results} before concluding in Section~\ref{sec:conclusion}.

\textbf{Notation: }
We will follow the convention of denoting vectors by boldface lower case  and matrices in capital boldface letters. The superscripts $(\cdot)^{\dag}$, $\|\cdot\|$, $\Trace(\cdot)$ and $(\cdot)^{-1}$, denote the conjugate transpose, Euclidean norm, the trace of a matrix, and the matrix inverse, respectively. $\E{x}$ stands for the expectation of random variable $x$ and $\vI_{M}$ is the identity matrix of size $M$. A circularly symmetric complex Gaussian random variable $x$ with mean  $\mu$ and variance $\sigma^2$ is represented as $\mathcal{CN}(\mu,\sigma^2)$. $\Ei(\cdot)$ is the exponential integral~\cite[Eq. (8.211.1)]{Integral:Series:Ryzhik:1992}. ${}_{2}F_{1}(\cdot,\cdot;\cdot;\cdot)$ is the Gauss hypergeometric function~\cite[Eq. (9.111)]{Integral:Series:Ryzhik:1992} and $F_1(\cdot;\cdot,\cdot;\cdot;\cdot,\cdot)$ is the Appell hypergeometric function~\cite[Eq. (5.8.2)]{Transcendental:book}. \small{$G_{p q}^{m n} \left( z \  \vert \  {a_1\cdots a_p \atop b_1\cdots b_q} \right)$ }\normalsize is the Meijer G-function~\cite[ Eq. (9.301)]{Integral:Series:Ryzhik:1992} and $D_{-1}(\cdot)$ is Parabolic cylinder function~\cite[Eq. (9.241.2)]{Integral:Series:Ryzhik:1992}.

\section{System Model}\label{sec:system model and assumption}
Consider a single cell wireless system with an AP, where  data from users in the UL channel, and data to the users in the DL channel are transmitted and received at the same time on the same frequency as shown in Fig.~\ref{fig: system model}. All users in the cell are located in a circular area with radius $R_c$ and the AP is located at the center. We assume that users are equipped with a single antenna, while the AP is equipped with $\nr$ receive and $\nt$ transmit antennas for FD operation.  The single antenna assumption is made for several pragmatic reasons. First, most mobile handsets are single antenna devices.  Second, since in the case of multiple antennas, the capacity is unknown or at least it will be a complicated optimization problem. Third, single antenna user equipment is also an exceedingly  common assumption made in massive MIMO  and other wireless literature. Also since we assume multiple-antenna  AP, it can cancel LI and provide a good rate for the UL / DL user etc.  In the sequel, we use subscript-$\Us$ for the UL user, subscript-$\Ds$ for the DL user, and subscript-$\AP$ for the AP. Similarly, we will use subscript-$\AP\AP$, subscript-$\AP{\Ds}$, subscript-${\Us}{\Ds}$, and subscript-${\Us}\AP$ to denote the AP-to-AP, AP-to-DL user, UL user-to-DL user, and UL user-to-AP channels, respectively.

We model the locations of the DL users $x_{\mathsf{d}}$ inside the disk as an independent two-dimensional homogeneous PPP $\Phi_{\mathsf{d}}=\{x_{\mathsf{d}}\}$ with spatial density $\lambda_{\mathsf{d}}$. The AP selects a DL user that is physically
nearest to it as well as an UL user $d$ distance away\footnote{The link distance $d$ can also be random without affecting the
main conclusions, since we can always derive the results by first conditioning on $d$ and then averaging over $d$. The fixed inter user distance assumption can be shown to preserve the integrity of conclusions even with random transmit distances.} from the DL user in
a random direction of angle $\theta$\footnote{Since UL/DL users simultaneously use the same channel resources (due to FD operation), a separation distance between the users is required in order to avoid inter-user interference effects. This distance is ensured by an appropriate scheduling algorithm, which selects UL/DL users with a distance higher than $d$. In our setup, we study a worst-case scenario where users are located in the smallest allowed distance as this scenario serves as a useful guideline for practical FD network design. Parametrization in terms of the inter user UL/DL distance has also been
adopted by some of the existing papers~\cite{Panwar:ICC14,Duplo:Techrep}.}~\cite{Haenggi:arXiv,Venkateswaran:Infocom:2015}. The AP selects a DL user that is physically nearest to it. We use the terms ``nearest DL user'' and ``scheduled DL user'' interchangeably throughout the paper to refer to this user. Selection of a nearest user is necessary for an FD AP    since transmitting very high power signals towards distant periphery users in order to guarantee a quality-of-service can cause overwhelming LI at the receive side of the AP~\cite{Everett:JWCOM:2014}. Moreover, cell sizes have been shrinking progressively over generations of network evolution. Therefore in some next generation networks each user will be in the coverage area of an AP and can be considered as a nearest user~\cite{Andrews:MCOM:2013}.
As a benchmark comparison we also consider the random user selection (RUS) in Section~\ref{sec:Numerical results}. Under RUS the AP randomly selects one of all candidate DL users with equal probability.

For a more realistic propagation  model, we  assume that the links  experience both large-scale path loss effects and small-scale Rayleigh fading. For the large-scale path loss, we assume the standard singular path loss model, $\ell(x,y)=\|x-y\|^{-\alpha}$, where $\alpha \geq 2$ denotes the path loss exponent and $\|x-y\|$ is the Euclidean distance between two nodes. If $y$ is at the origin, the index $y$ will be omitted, i.e., $\ell(x,0)=\ell(x)$.
In order to facilitate the analysis, we now set up a polar coordinate system in which the origin is at the AP and the scheduled DL user is at $x_{\Ds}=(r,0)$. Therefore, we have $\ell(x_{\Us}) = (r^2+d^2-2rd\cos\theta)^{-\alpha/2}$.
In the following, we will need the exact knowledge of the spatial distribution of the $\ell( x_{\Us})$ in terms of $r$ and $\theta$. Since we assume that nearest DL user is scheduled for downlink transmission, $x_{\Ds}$ denotes the distance between the AP and the nearest DL user. Therefore, the probability distribution function (pdf) of the nearest distance $x_{\Ds}$ for the homogeneous PPP $\Phi_{\Ds}$ with intensity $\lambda_{\Ds}$ is given by~\cite{Haenggi:IT:2005}.
\begin{align}\label{eq:pdf of the nearst distance}
 f_{r}(r)= 2\pi\lambda_{\Ds}r e^{-\lambda_{\Ds}\pi r^2},~\qquad  0 \leq r < \infty.
\end{align}
Moreover, angular distribution is uniformly distributed over $[0~2\pi]$ i.e., $f_{\theta}(\theta )=1/ 2\pi$.
\begin{figure}[t]
\centering
\includegraphics[width=75mm, height=75mm]{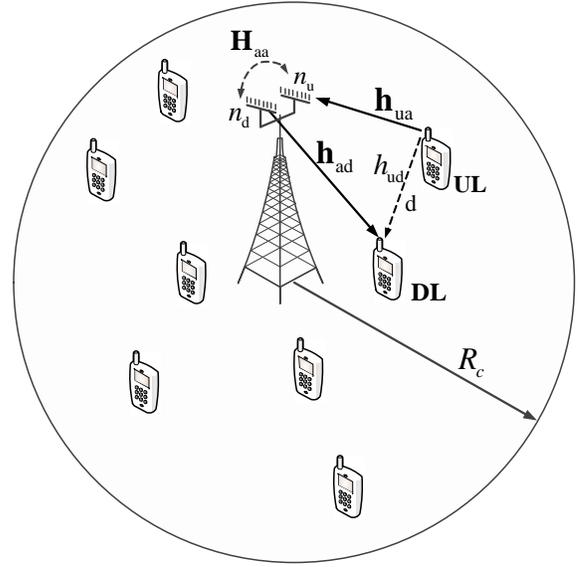}
\caption{System model for the FD wireless system. The AP communicates with single antenna HD spatially random user terminals to support simultaneous UL and DL transmissions in the presence of LI and internode interference. }
\label{fig: system model}
\vspace{-1em}
\end{figure}
\subsection{Signal Model}
We assume that the AP transmits with power $P_{\AP}$ and let $s_{\AP}$ be its transmitted data symbol with $\E{| s_{\AP} |^2}=1$. The transmitted data symbol $s_{\AP}$ is mapped onto the antenna array elements by the beamforming  $\WT\in \mathcal{C}^{\nt\times 1}$ with $\parallel\WT\parallel=1$. It is worth mentioning that the scheduled UL user, located at $x_{\Us}$, is served by receive antennas from AP at the same time, and it lacks coordination with concurrent active DL users. Therefore, the received signal for the scheduled DL user is given by
\begin{align}\label{eq:received signal at DL user}
y_{\Ds} = \sqrt{P_{\AP}\ell(x_{\Ds})}\vh_{\AP\Ds}\WT s_{\AP} + \sqrt{P_{\Us}\ell( x_{\Us},x_{\Ds})}h_{ \Us\Ds}s_{\Us} + z_{\Ds},
\end{align}
where $\vh_{\AP\Ds}\in \mathcal{C}^{1\times\nt}$ denotes the small-scale fading for the link between the AP and the active DL user. The entries of $\vh_{\AP\Ds}$ follow identically and independently distributed (i.i.d.) $\mathcal{CN}(0 , 1)$. $P_{\Us}$ is the transmit power of the UL user, $h_{ \Us\Ds}$ denotes the channel for the UL-DL user link, $s_{\Us}$ is the source symbol satisfying $\E{| s_{\Us} |^2}=1$. $z_{\Ds}$ is the additive white Gaussian noise (AWGN) at the DL user receiver with $\E{z_{\Ds} z_{\Ds}^*}=\Sn$.

From~\eqref{eq:received signal at DL user} the receive signal-to-interference-plus-noise ratio (SINR) of the scheduled DL user is given by
\begin{align}\label{eq:SINR: downlonk user}
\SINRd=\frac{P_{\AP} \ell(x_{\Ds}) \|\vh_{\AP\Ds}\WT\|^2 }
{P_{\Us} g_{ \Us\Ds} \ell( x_{\Us},x_{\Ds}) + \Sn},
\end{align}
where $g_{ \Us\Ds} = |h_{ \Us\Ds}|^2$ is the channel gain for the link between the UL and DL user.

On the other hand, with the linear receiver vector $\WR \in \mathcal{C}^{1 \times\nr}$, the received signal at the AP is given by
\vspace{-0.2em}
\begin{align}\label{eq:received signal at AP}
y_{\AP} = \sqrt{P_{\Us}\ell(x_{\Us})}\WR\vh_{\Us\AP} s_{\Us} +
\sqrt{P_{\AP}}\WR\vH_{\AP\AP} \WT s_{\AP} + \WR \vz_{\AP},
\end{align}
where $\vh_{\Us\AP}\in \mathcal{C}^{\nr\times 1}$ denotes the small-scale fading between the scheduled UL user and AP, $\vH_{\AP\AP} \in \mathcal{C}^{\nr \times \nt}$ is the channel matrix between the transmit and receive arrays which represents the LI~\cite{Riihonen:JSP:2011}. We model the LI channel with Rayleigh flat fading which is a well accepted model in the literature~\cite{Riihonen:JSP:2011,Riihonen:WCOm:2011}. In this model, the strong line-of-sight component of the LI channel is estimated and removed during the cancellation process employed at the AP. Therefore, the residual interference is mainly affected by the Rayleigh fading component of the loopback channel and its strength is proportional to the level of suppression achieved by the adopted specific cancellation method. Since each implementation of a particular analog/digital LI cancellation scheme can be characterized by a specific residual power, the elements of $\vH_{\AP\AP}$ can be modeled as i.i.d. $\mathcal{CN}(0,\Sap)$ random variables (RVs). This parameterization allows these effects to be studied in a generic way. Also,  $\vz_{\AP}$ is AWGN vector at the AP with $\E{\vz_{\AP} \vz_{\AP}^\dag}=\Sn \vI$.

Invoking~\eqref{eq:received signal at AP}, the resulting SINR expression at the AP can be computed as
\vspace{-0.3em}
\begin{align}\label{eq:SINR at AP in uplink}
\SINRAP =
\frac{P_{\Us}   \ell(x_{\Us} ) \| \WR\vh_{\Us\AP}\|^2}
{P_{\AP} \|\WR\vH_{\AP\AP} \WT\|^2 + \Sn\|\WR\|^2}.
\end{align}

In the next section, we consider different schemes for $\WT$ and $\WR$ and characterize the system performance using the UL and DL average sum rate given by
\vspace{-0.2em}
\begin{align}
\RFD  = {R}_{\AP} + {R}_{\Ds},\label{eq:achievable rate FD}
\end{align}
where ${R}_{\AP}=\E{ \log_2 \left(1+\SINRAP\right)}$ and  ${R}_{\Ds}=\E{ \log_2 \left(1+\SINRd\right)}$ are the spatial
average rate of the UL ($x_{\Us}\rightarrow \mathrm{AP }$) and DL ($\mathrm{AP} \rightarrow x_{\Ds}$), respectively. Specifically, we analyze the following schemes:
\begin{enumerate}
 \item The AP adopts MRC and MRT processing on receive and transmit antenna signals, respectively ($\MRCMRT$),
  \item The AP adopts MRC and ZF processing on receive and transmit antenna signals, respectively ($\MRCZF$),
 \item The AP employs ZF and MRT processing on receive and transmit antenna signals, respectively ($\ZFMRT$)\footnote{ Joint optimization of receive and transmit precoding vectors under the condition $\vw_{r}\vH_{\AP\AP}\vw_{t}=0$ is another ZF approach that can be studied. However, we leave it as future work since already $\MRCZF$ or $\ZFMRT$ can eliminate LI.}.
\end{enumerate}
We finally address the problem of joint optimal design of transmit and receive precoders at the AP for maximizing the average UL and DL sum rate.
\section{Performance Analysis} \label{sec:Performance Analysis}
In this section, the UL and DL average sum rates provisioned under different transmit-receive precoding schemes are evaluated and new expressions are derived\footnote{Note that the rate achieved by a DL user outside the cell radius $R_c$ is zero.}. We then consider the more general problem of joint transmit-receive precoder design at the AP for the sum-rate maximization problem.

Before proceeding into the derivation of the average sum rates for the specific schemes, it is useful to note that for a nonnegative random variable $X$, since $\E{X} = \int_{t=0}^{\infty}\Prob(X>t)dt$,  the average achievable rate can be written as
\vspace{-0.2em}
\begin{align}
R_i
   &=\int_{0}^{\infty}\left(1-F_{\SINRi}(\epsilon_t)\right)dt,\label{eq:achievable downlink rate}
\end{align}
where $F_{\SINRi}(z) =1-\Prob(\SINRi \geq z)$ is the cdf of $\SINRi$ with  $i\in\{\AP, \Ds\}$, and $\epsilon_t=2^t-1$. Therefore, the average achievable rate can be calculated from the cdf of the corresponding SINR.
\vspace{-1.5em}
\subsection{$\MRCMRT$ Scheme}
MRC processing for UL signals combined with MRT processing for the DL signals ($\MRCMRT$) is the optimal transmit-receive diversity technique
in the sense that it maximizes the SNR~\cite{TitusMRT:JCOM:1999}. Although $\MRCMRT$ processing is not optimal  in presence of LI, it is favored in  HD systems, because of its ability to balance the performance and system complexity.

With MRT processing, precoding vector is given by $\WT^{\MRT} =\frac{\vh_{\AP\Ds}^{\dag}}{ \parallel \vh_{\AP\Ds} \parallel}$. The MRC processing combines all the signals received from each antenna at the receive side of the AP. Mathematically, the combining vector can be expressed as $\WR^{\MRC} = \frac{\vh_{\Us\AP}^{\dag}}{\parallel \vh_{\Us\AP} \parallel}$. Therefore, the corresponding SINRs for the  $\text{MRC}_{\mathsf{(rx)}}/\text{MRT}_{\mathsf{(tx)}}$ scheme becomes
\vspace{-0.4em}
\begin{subequations}
\begin{align}
\SINRAP &=\frac{ P_{\Us}\ell(x_{\Us})\| \vh_{\Us\AP}\|^2}
{P_{\AP} \frac{\|\vh_{\Us\AP}^{\dag}\vH_{\AP\AP} \vh_{\AP\Ds}^{\dag}\|^2}{\|\vh_{\Us\AP}\vh_{\AP\Ds}\|}+ \Sn},\label{eq:SINRAP MRC/MRT}\\
\SINRd  &= \frac{P_{\AP} \ell(x_{\Ds}) \|\vh_{\AP\Ds}^{\dag}\|^2}
{P_{\Us} g_{ \Us\Ds} \ell( x_{\Us},x_{\Ds}) + \Sn}.\label{eq:SINRd MRC/MRT}
\end{align}
\end{subequations}

We now proceed to derive the cdfs of the SINRs given in~\eqref{eq:SINRAP MRC/MRT} and~\eqref{eq:SINRAP MRC/MRT} so that the spatial capacity can be evaluated using~\eqref{eq:achievable downlink rate}.
\newline
\textbf{\emph{Evaluation of $\boldsymbol{R_{\AP}}$}:}
In order to determine the cdf of $\SINRAP$ in~\eqref{eq:SINRAP MRC/MRT} the
distribution of the RV $\|\WR^{\MRC}\vH_{\AP\AP} \WT^{\MRT}\|^2 =\frac{\|\vh_{\Us\AP}^{\dag}\vH_{\AP\AP} \vh_{\AP\Ds}^{\dag}\|^2}{\|\vh_{\Us\AP}\vh_{\AP\Ds}\|}$ needs to be found. Note that the exact cdf expression of this RV valid for any $\nt /\nr$ is cumbersome to obtain in closed-form. As such, we consider now several cases as follows:

\emph{Case-1) ${\nr=1, \nt \geq 1}$:}
In this case the numerator in~\eqref{eq:SINRAP MRC/MRT} reduces to $P_{\Us}\ell(x_{\Us}) g_{\Us\AP}$ where $g_{\Us\AP} \triangleq |h_{\Us\AP}|^2$ is an exponential distributed RV.
Before proceeding, we need the cdf of $P_{\AP} \|\vh_{\AP\AP} \WT^{\MRT}\|^2 $, which can be written as
\vspace{-0.3em}
\begin{align}\label{eq:proof of the exp RV MRC/MRT case-1}
P_{\AP} \|\vh_{\AP\AP} \WT^{\MRT} \|^2 &= P_{\AP}(\vh_{\AP\AP} \WT^{\MRT} \WT^{\dag\MRT}\vh_{\AP\AP}^{\dag} )  \nonumber\\
&=P_{\AP} (\vh_{\AP\AP} \boldsymbol{\Phi}_t \diag\{1,0,\cdots,0\} \boldsymbol{\Phi}^{\dag}_t\vh_{\AP\AP}^{\dag} )  \nonumber\\
&=P_{\AP} |\hat{\vh}_{\AP\AP,1}|^2,
\end{align}
where $\vh_{\AP\AP} \in \mathcal{C}^{1 \times \nt}$, $\boldsymbol{\Phi}_t$ is an unitary matrix and $\hat{\vh}_{\AP\AP} = \vh_{\AP\AP}\boldsymbol{\Phi}_t$.  In~\eqref{eq:proof of the exp RV MRC/MRT case-1}, the second equality is due to the eigen-decomposition\footnote{ Note that $\WT$ is a $n_{\Ds} \times 1$ normalized column vector and has rank $1$.}. Let us denote $g_{\AP\AP}\triangleq|\hat{\vh}_{\AP\AP,1}|^2$. Since the elements of $\vh_{\AP\AP}$ are i.i.d. $\mathcal{CN}(0,\Sap)$ RVs, $g_{\AP\AP}$ is an exponential RV with parameter $\Sap$. Therefore, the denominator in (8a) depends of the LI channel $\vh_{\AP\AP}$. On the contrary the numerator of~\eqref{eq:SINRAP MRC/MRT} only depends on the UL user to AP link which is independent of the LI channel. Hence, the numerator and the denominator in~\eqref{eq:SINRAP MRC/MRT} are independent. Therefore the $\SINRAP$ cdf can be written as
\vspace{-0.3em}
\begin{align}\label{eq:cumulative distribution function of SINRu}
&F_{\SINRAP}(z)
= 1\!- \! \Erth{\Prob \left(g_{\Us\AP}\! \geq\! \frac{z}{P_{\Us} \ell(x_{\Us})} [P_{\AP} g_{\AP\AP} \! + \!\Sn] \Big| g_{\AP\AP}\right )},
\nonumber\\
&~=1\!-\!\Erth{
\frac{e^{-z\frac{\Sn}{P_{\Us}}(r^2+d^2-2rd\cos\theta)^{\alpha/2} }}
{1 \!+\!z\frac{P_{\AP}}{P_{\Us}}\Sap     (r^2+d^2-2rd\cos\theta)^{\alpha/2} }}\!.
 \end{align}
With the aid of the pdfs for $r$ and $\theta$ in~\eqref{eq:pdf of the nearst distance}, we can express $F_{\SINRAP}(z)$ as
\vspace{-0.3em}
\begin{align}\label{eq: cdf of SINRd polar coordinates general}
&F_{\SINRAP}(z)
=1-\lambda_{\Ds}\\
&\quad\times\int_{0}^{R_c}
\int_{0}^{2\pi}\frac{r e^{-\lambda_{\Ds}\pi r^2}e^{-\frac{z\Sn}{P_{\Us}} (r^2+d^2-2rd\cos\theta)^{\frac{\alpha}{2}}}}
{1 +z \frac{P_{\AP}}{P_{\Us}} \Sap(r^2+d^2-2rd\cos\theta)^{\frac{\alpha}{2}}}d\theta dr.\nonumber
\end{align}
By making the variable change $\theta = \arccos\left(\frac{r^2+d^2-y^2}{2rd}\right)$ the integrand of~\eqref{eq: cdf of SINRd polar coordinates general} for $\alpha=2n$, $n=1,2,\cdots$ can be written as $f(y) = \frac{1} {(1 + z \frac{P_{\AP}}{P_{\Us}} \Sap y^n)\sqrt{y - (d-r)^2} \sqrt{(d+r)^2 -y}}$. To proceed, we
convert the integral over $y$ into a complex integral over $\mathcal{Z}$. By considering a positively oriented simple closed contour $C$ such that $f(\mathcal{Z})$ is analytic everywhere in the finite plane except for a finite number of singular points interior to $C$ and using the fact $\int_{C} f(\mathcal{Z}) d\mathcal{Z} = 2\pi i \res_{\mathcal{Z}=0}\left[\frac{1}{\mathcal{Z}^2} f\left(\frac{1}{\mathcal{Z}}\right)\right]$ where $\res_{\mathcal{Z}=\mathcal{Z}_k}(\cdot)$ represents the residue at $\mathcal{Z}_k$, after some manipulations,~\eqref{eq: cdf of SINRd polar coordinates general} can be expressed as
\vspace{-0.3em}
\begin{align}\label{eq:cdf:1-fold}
F_{\SINRAP}(\gamma_{th})
&\! =\!1 \!-\! \lambda_\Ds\pi \sum_{k=1}^{n}\int_{0}^{R_c}\!\!\! r e^{-\lambda_{\Ds}\pi r^2} \res_{\mathcal{Z}=\mathcal{Z}_k} (f(\mathcal{Z})) dr,
\end{align}
where  $\res_{\mathcal{Z}=\mathcal{Z}_k} (f(\mathcal{Z})) = \lim_{\mathcal{Z}\rightarrow \mathcal{Z}_k} (\mathcal{Z}-\mathcal{Z}_k)f(\mathcal{Z})$ and $\mathcal{Z}_k = (\frac{P_{\Us}}{P_{\AP}} \frac{1}{\gamma_{th}\Sap})^{\frac{1}{n}} e^{-i\frac{(2k+1)\pi}{n}}$. In general, the double integral in~\eqref{eq:cdf:1-fold} does not admit a simple analytical solution for an arbitrary value of $\alpha$. However, the cdf can be straightforwardly  evaluated using numerical integration. By substituting \eqref{eq: cdf of SINRd polar coordinates general} into~\eqref{eq:achievable downlink rate}, after some manipulations, the exact average capacity of the UL user can be written as
\vspace{-0.4em}
\begin{align}\label{eq:spatial average capacity of the downlink user}
R_{\AP}&=
\frac{\lambda_{\Ds}}{\ln2}\int_{0}^{R_c}\int_{0}^{2\pi}
\frac{r e^{-\lambda_{\Ds}\pi r^2}}
{1 -\frac{P_{\AP}}{P_{\Us}}\Sap    (r^2+d^2-2rd\cos\theta)^{\alpha/2} }\nonumber\\
 &\quad\times\left(e^{\frac{\Sn}{P_{\AP}\Sap}} \Ei\left(-\frac{\Sn}{P_{\AP}\Sap}\right)
-e^{\frac{\Sn}{P_{\Us}}(r^2+d^2-2rd\cos\theta)^{\alpha/2}}\right.\nonumber\\
&\quad\left.\times
\Ei\left(-\frac{\Sn}{P_{\Us}}(r^2+d^2-2rd\cos\theta)^{\alpha/2}\right)
\right)d\theta dr.
 \end{align}

The following propositions characterize $R_{\AP}$ for the interference-limited scenario with $\Sn=0$ and the special cases with\footnote{Note that $\alpha=2$ and $\alpha=4$ correspond to free space propagation and typical rural areas, respectively, and constitute useful bounds for practical propagation conditions} $\alpha=2$ and $\alpha=4$.
\begin{proposition}\label{Propos:average capacity of the UL user MRC/MRT}
The spatial average rate of the UL user in Case-1 for $\alpha=2$ is given by
\vspace{-0.4em}
\begin{align}\label{eq: acheivable rate UL: alpha 2 Final}
&R_{\AP}\!= \!\frac{4 }{\ln2}\frac{P_{\Us}}{\Sap P_{\AP }}\!\!
\sum_{k=0}^{\infty}\!\!
\frac{(-1)^k(2\pi\lambda_{\Ds})^{k+1}}
{\Gamma(k+1)}
\!\!\!\int_0^{\infty}\!\!\!\!\!\!
\frac{ c^{k+\frac{1}{2}}}{z(z+1)}
\!\left(\!\frac{b\!-\!\sqrt{c}\varrho}{c\! -\! b^2}
\!\right)^{\!\!k+1}\nonumber\\
&~\times
F_{1}\left(k+1;k+1,k+1;k+2;\frac{b-\sqrt{c}\varrho}{b+\sqrt{c}},\frac{b-\sqrt{c}\varrho}{b-\sqrt{c}}\right)dz.
\end{align}
where $\Gamma(\cdot)$ is the Gamma function~\cite[Eq. (8.310.1)]{Integral:Series:Ryzhik:1992}. Moreover, $c =\left(\frac{P_{\Us}}{P_{\AP}} \frac{ 1}{z\Sap}+d^2\right)^2 $, $b =\frac{P_{\Us}}{P_{\AP}} \frac{ 1}{z\Sap}-d^2$, and  $\varrho=(\sqrt{R_c^4+bR_c^2+c}-\sqrt{c})/R_c^2$.
\end{proposition}

\begin{proof}
To prove this proposition, the following lemma is useful. The proof of Lemma~\ref{propos: cdf for SINRd uplink MRC/MRT case1 alpha2} is  presented in Appendix~\ref{APX: propos: cdf for SINRd uplink MRC/MRT case1 alpha2 }.
\begin{Lemma}\label{propos: cdf for SINRd uplink MRC/MRT case1 alpha2}
The cdf of $\SINRAP$, for $\alpha=2$ is given by
\begin{align}\label{eq: cdf of SINRd integral over r: alpha 2 Final}
&F_{\SINRAP}(z)
=1\!-\!\frac{P_{\Us}}{P_{\AP }} \frac{ 8\pi\lambda_{\Ds} }{z\Sap}
~\sum_{k=0}^{\infty}\frac{(-2\pi\lambda_{\Ds} c)^k}{\Gamma(k+1)}
\sqrt{c}\left(\frac{b-\sqrt{c}\varrho}{c - b^2}
\right)^{k+1}\nonumber\\
&~\times
F_{1}\left(\!k+1;k+1,k+1;k+2;\frac{b\!-\!\sqrt{c}\varrho}{b\!+\!\sqrt{c}},\frac{b\!-\!\sqrt{c}\varrho}{b\!-\!\sqrt{c}}\!\right)\!.
\end{align}
\end{Lemma}

Next, by using Lemma~\ref{propos: cdf for SINRd uplink MRC/MRT case1 alpha2}, and plugging ~\eqref{eq: cdf of SINRd integral over r: alpha 2 Final} into~\eqref{eq:achievable downlink rate}, after some algebraic manipulation, the desired result in~\eqref{eq: acheivable rate UL: alpha 2 Final} can be obtained.
\end{proof}

Before proceeding further, we present the following lemma, which will be used to establish an upper bound on the achievable rate of the UL user for $\alpha=4$.
\begin{Lemma}\label{lemma:cdf for SINRd uplink MRC/MRT case1 alpha4}
The cdf of $\SINRAP$ for $\alpha=4$ is lower bounded as
\begin{align}\label{eq: proof of the cdf of SINRd alpha final}
F_{\SINRAP}(z)
&> 1-
 \sum_{k=0}^{\infty}\frac{(-1)^k(\lambda_{\Ds}\pi R_c^2)^{k+1}}{\Gamma(k+2)}\nonumber\\
&~\times
 {}_{2}F_{1} \left(\!1,\frac{k\!+\!1}{2},\frac{k\!+\!1}{2}+1,\!-z\Sap\frac{ P_{\AP}}{P_{\Us}}R_c^4\!\right).
\end{align}
\end{Lemma}
\begin{proof}
See Appendix~\ref{APX: propos: cdf for SINRd uplink MRC/MRT case1 alpha4 }.
\end{proof}

\begin{proposition}\label{Prop: Acheivable rate uplink alpha4}
For $\alpha=4$, the spatial average rate of the UL user is upper bounded by
\begin{align}\label{eq:achievable uplink rates:special case:proposition}
R_{\AP}&< \frac{1}{2\log 2}\sum_{k=0}^{\infty}\frac{(-1)^k(\lambda_{\Ds}\pi R_c^2)^{k+1}}{\Gamma(k+1)}\nonumber\\
&~\times
G_{3 ~3}^{3~ 2}
\left(\frac{P_{\Us}}{ P_{\AP}}\frac{1}{R_c^4\Sap} \bigg\vert {1, 1, 1+\frac{k+1}{2} \atop 1, 1, \frac{k+1}{2}} \right).
\end{align}
\end{proposition}
\begin{proof}
By substituting the lower bound of $F_{\SINRAP}(\cdot)$ from Lemma~\ref{lemma:cdf for SINRd uplink MRC/MRT case1 alpha4} into~\eqref{eq:achievable downlink rate}, and applying the transformation $y=2^t-1$, an upper bound for the average rate of the UL user can be expressed as
\begin{align}
&R_{\AP}< \frac{1}{\ln2}\sum_{k=0}^{\infty}\frac{(-1)^k(\lambda_{\Ds}\pi R_c^2)^{k+1}}{\Gamma(k+2)}\nonumber\\
&~\times\underbrace{\int_0^{\infty} \frac{1}{y+1}~{}_{2}F_{1} \left(1,\frac{k+1}{2},\frac{k+1}{2}+1,\!-\frac{ P_{\AP}}{P_{\Us}}R_c^4\Sap y\!\right)dy,}_{\mathcal{I}_1}\nonumber
\end{align}
where the integral, $\mathcal{I}_1$ can be expressed~\cite[Eq. (17)]{Adamchik:1990} in terms of the tabulated Meijer
G-function as
\vspace{-0.2em}
\begin{align}
\mathcal{I}_1\!=\! \frac{2}{k+1}&\int_0^{\infty}
\!\!G_{1 1}^{1 1} \left( y \ \Big \vert \  {0 \atop 0} \right)\\
&~\times
G_{2 2}^{1 2} \left( \frac{P_{\AP}}{P_{\Us}}R_c^4\Sap y\! \  \Big\vert \  {0, 1-\frac{k+1}{2} \atop 0,-\frac{k+1}{2}} \right)dy.\nonumber
\end{align}
The above integral can be solved with the help of~\cite[Eq. (21)]{Adamchik:1990} and~\cite[Eq. (8.2.2.14)]{Prudnikov:vol3} to yield the desired result in~\eqref{eq:achievable uplink rates:special case:proposition}.
\end{proof}
Note that the Meijer G-function is available as
a built-in function in many mathematical software packages, such as Maple and Mathematica.

\emph{Case-2) $\nt=1, \nr>1$:} In this case, the cdf of $\SINRAP$ can be expressed as
\begin{align}\label{eq:SINRa nt=1}
F_{\SINRAP}(z)
& = \Erth{\Prob\left(\frac{\gamma_{\Us\AP}}{\gamma_{\AP\AP} + \Sn} \leq z\right) }\\
&= 1- \Erth{ \int_{0}^{\infty}F_{\gamma_{\AP\AP}}\left(\frac{x}{z}\right)  f_{\gamma_{\Us\AP}} (z+x) dx},\nonumber
\end{align}
where $\gamma_{\Us\AP} =P_{\Us}\ell(x_{\Us})\parallel\vh_{\Us\AP}\parallel^2 $ and $\gamma_{\AP\AP} = P_{\AP}\frac{\|\vh_{\Us\AP}^{\dag}\vh_{\AP\AP} \vh_{\AP\Ds}^{\dag}\|^2}{\|\vh_{\Us\AP}\vh_{\AP\Ds}\|}$ with $\vh_{\AP\AP} \in \mathcal{C}^{\nr\times1}$. Moreover, $F_{\gamma_{\AP\AP}}(\cdot)$ and $f_{\gamma_{\Us\AP}} (\cdot)$ denote the cdf and the pdf of the $\gamma_{\AP\AP}$ and $\gamma_{\Us\AP}$, respectively. In the sequel, we first derive the expressions for $F_{\gamma_{\AP\AP}}(\cdot)$ and $f_{\gamma_{\Us\AP}} (\cdot)$, then use these expressions to derive the $F_{\SINRAP}(z)$.

It is easy to show that $\gamma_{\Us\AP}$ follows central chi-square distribution with $2\nr$ degrees-of-freedom\footnote{
In what follows, we will use the notation $x\sim\chi_{2K}^2$ to denote $x$ that is chi-square distributed with $2K$ degrees-of-freedom.} whose pdf is given by~\cite{MathematicalStatistics_1978}
\begin{align}\label{eq:pdf of gamma_ua: chi2}
f_{\gamma_{\Us\AP}} (x) = \frac{x^{n_{\Us}-1}}{ (P_{\Us}   \ell(x_{\Us} ))^{n_{\Us}}\Gamma(n_{\Us})} e^{-\frac{x}{P_{\Us} \ell(x_{\Us} )}}.
\end{align}
Before deriving the cdf of $\gamma_{\AP\AP}$, we first note that $\gamma_{\AP\AP}$ can be written as
\begin{align}\label{eq:proof of the exp RV}
\gamma_{\AP\AP} &= P_{\AP} (\WR^{\dag \MRC}\vh_{\AP\AP}\vh_{\AP\AP}^{\dag}\WR^{\MRC} ) , \nonumber\\
&=P_{\AP} (\WR^{\dag \MRC} \vU \diag\{\lambda_{max},0,\cdots,0\} \vU^{\dag}\WR^{\MRC} ),  \nonumber\\
&=\frac{|\vh_{\Us\AP,1}|^2}{\sum_{i=1}^{\nr} |\vh_{\Us\AP,i}|^2}P_{\AP} \lambda_{max},\nonumber\\
&\triangleq  YZ,
\end{align}
where $\vU$ is unitary matrix, $Y=\frac{|\vh_{\Us\AP,1}|^2}{\sum_{i=1}^{\nr} |\vh_{\Us\AP,i}|^2}$, $Z=P_{\AP} \lambda_{max}$, and $\lambda_{max}$ is the largest eigenvalue of the Wishart matrix $\vh_{\AP\AP}\vh_{\AP\AP}^{\dag}$. Here, the second equality is due to
the eigen-decomposition. It is well known that $\lambda_{max}\sim\chi_{2\nr}^2$~\cite{Xiaohu:Eigen:IT:2007} and $Y$ follows a beta distribution with shape parameters $1$ and $\nr-1$, which is denoted as $Y\sim \Bta(1,\nr-1)$~\cite{MathematicalStatistics_1978}. Accordingly, the cdf of $\gamma_{\AP\AP}$, can be found as
\begin{align}\label{eq:cdf of gamma_aa:chi2 beta}
&F_{\gamma_{\AP\AP}}(w)=\Prob(YZ < w),\nonumber\\
&~=\int_{0}^{1} F_{Z}\left(\frac{w}{y}\right) f_{Y}(y) dy,\nonumber\\
&~\stackrel{(a)}{=}\frac{1}{\Gamma(\nr-1)}\!\int_{1}^{\infty}\!\!\! x^{-\nr}(x-1)^{\nr-2}\gamma\left(\!\nr, \frac{\Sap}{P_{\AP}}wx\!\right) dx,\nonumber\\
&~\stackrel{(b)}{=}\frac{1}{\Gamma(\nr-1)}\!\int_{1}^{\infty}\!\!\! x^{-\nr}(x-1)^{\nr-2}
G_{1 2}^{1 1} \left(\! \frac{\Sap}{P_{\AP}}wx \  \Big\vert \  {1 \atop \nr, 0} \!\right)dx,\nonumber\\
&~\stackrel{(c)}{=}G_{2 3}^{2 1} \left( \frac{\Sap}{P_{\AP}}w \  \Big\vert \  {1, \nr \atop 1, \nr, 0} \right),
\end{align}
where $\gamma(\cdot,\cdot)$ is the lower incomplete Gamma function defined by~\cite[Eq. (8.350.1)]{Integral:Series:Ryzhik:1992}. In~\eqref{eq:cdf of gamma_aa:chi2 beta}, the equality (a) follows by substituting $x=1/y$, (b) is obtained by expressing the function $\gamma(\cdot,\cdot)$ in terms of Meijer G-function according to~\cite[Eq. (8.4.16.1)]{Prudnikov:vol3}, and (c) follows with the help of~\cite[Eq. (7.811.3)]{Integral:Series:Ryzhik:1992}.
\begin{Remark}
For $\nr=1$, using the equality~\cite[Eq. (8.2.2.18)]{Prudnikov:vol3} and applying \cite[Eq. (8.4.3.4)]{Prudnikov:vol3}, the cdf in~\eqref{eq:cdf of gamma_aa:chi2 beta} leads to a closed-form result for the cdf of the exponential RV with parameter $\Sap/P_{\AP}$, which confirms our analysis.
\end{Remark}

Now, by substituting~\eqref{eq:pdf of gamma_ua: chi2} and~\eqref{eq:cdf of gamma_aa:chi2 beta} into~\eqref{eq:SINRa nt=1}, we have
\begin{align}\label{eq:SINRa nt=1, cdf pdf}
F_{\SINRAP}(z)&=\!
 1-
 \frac{1}{\Gamma(n_{\Us})P_{\Us}^{\nr}}
 \Erth{ \frac{1}{(\ell(x_{\Us} ))^{n_{\Us}}}
 \int_{0}^{\infty}\!
 (x\!+\!z)^{n_{\Us}-1} \right.\nonumber\\
 &\left.\times e^{-\frac{x+z}{P_{\Us} \ell(x_{\Us} )}}
 G_{2 3}^{2 1} \left( \frac{x}{z}\frac{1}{P_{\AP}} \  \Big\vert \  {1, \nr \atop 1, \nr, 0} \right)
dx}.
\end{align}

Evaluation of~\eqref{eq:SINRa nt=1, cdf pdf} is difficult. In order to circumvent this challenge, a common approach adopted in the performance analysis literature is to neglect the AWGN term. Therefore, by using\cite[Eq. (2.24.3.1)]{Prudnikov:vol3}~\eqref{eq:SINRa nt=1, cdf pdf} can be evaluated as
\begin{align}\label{eq:SINRa nt=1, cdf pdf interference limited1}
F_{\SINRAP}(z)&=
 1-
  \frac{1}{\Gamma(n_{\Us})P_{\Us}^{\nr}}
 \Erth{ \frac{1}{( \ell(x_{\Us} ))^{n_{\Us}}}
 \int_{0}^{\infty}
 x^{n_{\Us}-1}\right.\nonumber\\
 &\quad\left.\times e^{-\frac{x}{P_{\Us} \ell(x_{\Us} )}}
 G_{2 3}^{2 1} \left( \frac{x}{z}\frac{1}{P_{\AP}} \  \Big\vert \  {1, \nr \atop 1, \nr, 0} \right)
 dx}\nonumber\\
 &=
 1\!-\! \frac{1}{\Gamma(n_{\Us})}
 \Erth{   G_{3 3}^{2 2} \left( \frac{P_{\Us}}{ P_{\AP}} \frac{\ell(x_{\Us} )}{z}\!\! \  \Big\vert \!\!\!\  {1-\nr, 1, \nr \atop 1, \nr, 0} \!\right)\!}\nonumber\\
 &=
 1- \frac{\lambda_{\Ds}}{\Gamma(n_{\Us})}
 \int_{0}^{R_c}\int_{0}^{2\pi}  r e^{-\lambda_{\Ds}\pi r^2} \nonumber\\
 &\quad\times G_{3 3}^{2 2} \left( \frac{P_{\Us}}{ P_{\AP}} \frac{\ell(x_{\Us} )}{z} \  \Big\vert \  {1-\nr, 1, \nr \atop 1, \nr, 0} \right)d\theta dr.
\end{align}
\begin{proposition}
The average rate of the UL user in the
interference-limited case can be expressed as
\begin{align}\label{eq:SINRa nt=1, cdf pdf interference limited}
&R_{\AP}=
  \frac{\lambda_{\Ds}}{\Gamma(n_{\Us})}
 \int_{0}^{\infty}\int_{0}^{R_c}\int_{0}^{2\pi}  r e^{-\lambda_{\Ds}\pi r^2} \\
 &~\times G_{3 3}^{2 2}
 \left( \frac{P_{\Us}}{ P_{\AP}} \frac{1}{\epsilon_t}(r^2+d^2-2rd\cos\theta)^{\frac{\alpha}{2}}\!\! \  \Big\vert \!\!\  {1-\nr, 1, \nr \atop 1, \nr, 0} \right)d\theta dr dt.\nonumber
\end{align}
\end{proposition}
\begin{proof}
The proof is straightforward and follows from the definition of $R_{\AP}$.
\end{proof}
\textbf{\emph{Evaluation of $\boldsymbol{R_{\Ds}}$}:} In order to derive a general expression for the $R_{\Ds}$, according to~\eqref{eq:achievable downlink rate}, we need  to obtain the cdf of $\SINRd$ in~\eqref{eq:SINRd MRC/MRT}, which can be written as
\begin{align}\label{eq:cdf of SINRd first step nd}
&F_{\SINRd}(z)
=1-\EI{\Prob\left(P_{\AP} \|\vh_{\AP\Ds}\|^2 r^{-\alpha}
 \geq z (I_{\Ds,\Us} + \Sn)\right)\big\vert r},\nonumber\\
&=1\!-\!\EI{1\!-\!\frac{1}{\Gamma(\nt)}\gamma\left(\nt,\frac{z}{P_{\AP}}r^{\alpha}\left(I_{\Ds,\Us} \!+\!\Sn\right)\right)\big\vert r}\!,
\end{align}
where $I_{\Ds,\Us} =P_{\Us} g_{ \Us\Ds} \ell( x_{\Us},x_{\Ds})$. Note that in our system model the randomness of the $I_{\Ds,\Us}$ is due to the fading power envelope $g_{\Us\Ds}$. As such, $F_{\SINRd}(z)$ can be re-expressed as
\vspace{-0.2em}
\begin{align}\label{eq:cdf of SINRd first step nd over Id}
&F_{\SINRd}(z)
=\frac{1}{\Gamma(\nt)}\\
 &~\times\Er{\int_{0}^{\infty}\gamma\left(\nt,\frac{z}{P_{\AP}}r^{\alpha} \left(P_{\Us}xd^{-\alpha} +\Sn\right)\right)e^{-x}dx}.\nonumber
\end{align}

Plugging \eqref{eq:cdf of SINRd first step nd over Id} into~\eqref{eq:achievable downlink rate} and using the identity~\cite[Eq. (8.356.3)]{Integral:Series:Ryzhik:1992}, exact average rate of the DL user can be written as
\vspace{-0.2em}
\begin{align}\label{eq: ache. sum rate dl user general}
R_{\Ds}
&=\frac{2\pi\lambda_{\Ds}}{\Gamma(\nt)}
\int_{0}^{\infty}
\int_{0}^{R_c}
\int_{0}^{\infty}r\Gamma\left(\nt,\frac{\epsilon_t}{P_{\AP}}r^{\alpha}\left(P_{\Us}xd^{-\alpha} +\Sn\right)\right)
\nonumber\\
 &\quad\times
 e^{-x}e^{-\lambda_{\Ds}\pi r^2}dxdrdt,
\end{align}
where $\Gamma(\cdot,\cdot)$ is upper incomplete Gamma function~\cite[Eq. (8.350.2)]{Integral:Series:Ryzhik:1992}.
Moreover, for  $\Sn=0$ (interference-limited case), using the cdf provided in the following lemma, Proposition~\ref{Prop: Acheivable rate downlink}  presents the average rate of the DL user.
\begin{Lemma}\label{lemma: cdf of SINRd general MRC/MRT}
The cdf of $\SINRd$, can be expressed as
\begin{align}\label{eq: cdf SINRd MRC/MRT}
F_{\SINRd}(z)
&=\frac{2}{\alpha}\sum_{k=0}^{\infty}\frac{(-1)^k (\lambda_{\Ds}\pi R_c^2)^{k+1}}{\Gamma(k+1)}\\
 &~\times
 G_{2, 2}^{1, 2} \left(\!\left(\frac{R_c}{d}\right)^{\alpha}\frac{P_{\Us}}{P_{\AP} }z \
 \Big\vert \ {  \nt, 1-\frac{2(k+1)}{\alpha} \atop 1, -\frac{2(k+1)}{\alpha}}\! \right).\nonumber
\end{align}
\end{Lemma}
\begin{proof}
For $\Sn=0$, with the help of~\cite[Eq. (6.451.1)]{Integral:Series:Ryzhik:1992},~\eqref{eq:cdf of SINRd first step nd over Id} can be written as
\begin{align}\label{eq:cdf of SINRd first step nd IL}
&F_{\SINRd}(z)
\!=\!2\pi\lambda_{\Ds}\int_{0}^{R_c}\!\!\!\!r
e^{-\lambda_{\Ds}\pi r^2}
\left(1 \!+\! \frac{1}{\left(\frac{r}{d}\right)^{\alpha}\frac{P_{\Us}}{P_{\AP} }z}\right)^{-\nt}\!\!dr.
\end{align}
Using the McLaurin series representation of the exponential function and expressing $(1 + [\left(\frac{r}{d}\right)^{\alpha}\frac{P_{\Us}}{P_{\AP} }z]^{-1})^{-\nt}$ as Meijer G-function with~\cite[Eq. (10)]{Adamchik:1990} and~\cite[Eq. (8.2.2.14)]{Prudnikov:vol3}, we can write
\begin{align}\label{eq: G-functions and series expanssion}
F_{\SINRd}(z)
&=2\sum_{k=0}^{\infty}\frac{(-1)^k (\lambda_{\Ds}\pi)^{k+1}}{\Gamma(k+1)}\\
 &~\times\int_{0}^{R_c}
r^{2k+1}
 G_{1 1}^{1 1} \left( \left(\frac{r}{d}\right)^{\alpha}\frac{P_{\Us}}{P_{\AP} }z \
 \Big\vert \  {\nt \atop 1} \right)
 dr.\nonumber
\end{align}
Now, using~\cite[Eq. (26)]{Adamchik:1990} the desired result in~\eqref{eq: cdf SINRd MRC/MRT} is obtained.
\end{proof}
In order to gain more insight into the system parameters,
    we study the special case of $\alpha=2$ which allows a closed-form solution, given by
    \begin{align}
    &F_{\SINRd}(z) = \Gamma(\nt+1)\Psi\left(\nt,0,\frac{P_{\AP}}{P_{\Us} }\frac{d^2}{z} \lambda_{\Ds}\pi\right),\nonumber
    \end{align}
    where $\Psi(a,b,z)$ is the Tricomi confluent hypergeometric function defined in~\cite[Eq. (9.211.4)]{Integral:Series:Ryzhik:1992}. By using the fact that $\Psi\left(a,0,z\right)\simeq 1/\Gamma(1+a)$, for small values of $z$, we get $F_{\SINRd}(z) =1$. This result is intuitive, since it indicates that the DL transmissions are in outage, when the UL user transmit power cause overwhelming internode interference at the DL user (i.e., $\frac{P_{\AP}}{P_{\Us} }\rightarrow0$) or/and internode distance is significantly decreased (i.e., $d\rightarrow 0$).

\begin{proposition}\label{Prop: Acheivable rate downlink}
The average rate of the DL user in the interference-limited case can be expressed as
\begin{align}\label{eq: achivable sum rate DL user general}
R_{\Ds}&=
 \frac{2}{\alpha\ln2}\left(\frac{R_c}{d}\right)^{\alpha}\frac{P_{\Us}}{P_{\AP} }
 \sum_{k=0}^{\infty}\frac{(-1)^k (\lambda_{\Ds}\pi R_c^2)^{k+1}}{\Gamma(k+1)}\nonumber\\
 &~\times
 G_{5 5}^{3 4} \left(\left(\frac{R_c}{d}\right)^{\alpha}\frac{P_{\Us}}{P_{\AP} } \
 \Big\vert \ {  0, \nt, 1-\frac{2(k+1)}{\alpha}, -1,0 \atop 1,-1,-1, -\frac{2(k+1)}{\alpha}, 1} \right).
\end{align}
\end{proposition}
\begin{proof}
The achievable rate $R_{\Ds}$ can be written as
\begin{align}\label{eq: achievable rate Rd in term of log}
R_{\Ds}&= \int_{0}^{\infty} \log(1 + z)f_{\SINRd}(z)dz,
\end{align}
where $f_{\SINRd}(z)$ denotes the pdf of the $\SINRd$ in~\eqref{eq:SINRd MRC/MRT}. Taking the first order derivative of~\eqref{eq: cdf SINRd MRC/MRT} with respect to $z$, using the expression for the derivative of the Meijer G-function, given in~\cite[Eq. (8.2.2.32)]{Prudnikov:vol3}, the pdf of $\SINRd$ can be obtained. By plugging the result into~\eqref{eq: achievable rate Rd in term of log} we get
\begin{align}
R_{\Ds}&=
 \frac{2}{\alpha\ln2}\left(\frac{R_c}{d}\right)^{\alpha}\frac{P_{\Us}}{P_{\AP} }
 \sum_{k=0}^{\infty}\frac{(-1)^k (\lambda_{\Ds}\pi R_c^2)^{k+1}}{\Gamma(k+1)} \mathcal{I}_{2}(z), \nonumber
\end{align}
where
\begin{align}
\mathcal{I}_{2}(z) &=
 \int_{0}^{\infty}z^{-1}\ln(1+z) \nonumber\\
 &~\times G_{3, 3}^{1, 3} \left(\left(\frac{R_c}{d}\right)^{\alpha}\frac{P_{\Us}}{P_{\AP} }z \
 \Big\vert \ {  0,\nt, 1-\frac{2(k+1)}{\alpha} \atop 1, -\frac{2(k+1)}{\alpha},1}\! \right)dz.\nonumber
\end{align}
The above integral can be evaluated by expressing $\ln(\cdot)$ in terms of a Meijer G-function~\cite[Eq. (8.4.6.5)]{Prudnikov:vol3} and then using~\cite[Eq. (21)]{Adamchik:1990} to yield~\eqref{eq: achivable sum rate DL user general}.
\end{proof}
\begin{Remark}\label{remark : new}
$R_{\Ds}$ is a monotonically increasing function of $\nt$. This can be explained by first noting that $Q(a,x) = \Gamma(a,x)/\Gamma(a)$  is a decreasing function of $x$, whose range is from $1$ (at $x=0$) to 0 (at $x\rightarrow \infty$). Also, it can be readily checked that $Q(a,x) = Q(a+1,x) - e^{-x}x^a/\Gamma(a+1)$ which results in $Q(a+1,x)>Q(a,x)$ for $x>0$. Combining these results together with the fact that $Q(a,x)$ has only one turning point at $x=a-1$, we conclude that the area under the curve of $Q(a,x)$ is increased when $a$ is increased.
\end{Remark}

It is worth noting that for the special case of $n_{\Ds}=1$, the numerator in~\eqref{eq:SINRd MRC/MRT} becomes exponentially distributed. Hence, by using the similar approach as in the proof of Lemma~\ref{lemma:cdf for SINRd uplink MRC/MRT case1 alpha4}, the cdf of $\SINRd$ can be written as
\vspace{-0.5em}
\begin{align}\label{eq: final cdf of SIR DL nt=1}
&F_{\SINRd}(z)
= 1-
 \sum_{k=0}^{\infty}\frac{(-1)^k(\lambda_{\Ds}\pi R_c^2)^{k+1}}{\Gamma(k+2)}\nonumber\\
 &~\times
{}_{2}F_{1} \left(1,\frac{2(k+1)}{\alpha},\frac{2(k+1)}{\alpha}+1,-z\frac{P_{\Us}}{ P_{\AP}}\left(\frac{R_c}{d}\right)^{\alpha}\right)\!.
\end{align}
Accordingly, substituting~\eqref{eq: final cdf of SIR DL nt=1} into~\eqref{eq:achievable downlink rate}, we can compute the average rate of the DL user, using the following Corollary.
\begin{Corollary}\label{Prop: Acheivable rate downlink nd=1}
For the special case of $n_{\Ds}=1$, the spatial average rate of the DL user in the interference-limited case can be expressed as
\vspace{-0.5em}
\begin{align}\label{eq:achievable sum rate of the DL}
R_{\Ds}&= \frac{\alpha}{2}\sum_{k=0}^{\infty}\frac{(-1)^k(\lambda_{\Ds}\pi R_c^2)^{k+1}}{(k+1)\Gamma(k+2)}\nonumber\\
 &~\times
  G_{3 ~3}^{3~ 2} \left(\frac{ P_{\AP}}{P_{\Us}}\left(\frac{d}{R_c}\right)^{\alpha}
  \bigg\vert {1, 1, 1+\frac{2(k+1)}{\alpha}  \atop 1,1+\frac{2(k+1)}{\alpha},1} \right).
\end{align}
\end{Corollary}
\begin{proof}
The proof, similar to \emph{Proposition~\ref{Prop: Acheivable rate uplink alpha4}} and is omitted.
\end{proof}
\textbf{\emph{\\Outage probability:}} As a byproduct of our analysis, the outage probability can be evaluated. The outage probability is an important quality-of-service metric defined
as the probability that $\SINRi$, $i\in\{\AP, \Ds\}$, drops below an acceptable SINR threshold, $\gamma_{\mathsf{th}}$. Mathematically, the outage probability can be obtained by evaluating the cdf of the received SINR at $\gamma_{\mathsf{th}}$. The following corollaries  establish the UL and DL user outage probability valid in the interference-limited case (i.e., $\Sn=0$).
\begin{Corollary}
For the $\MRCMRT$ scheme with ($\nr=1$, $\nt\geq1$), the UL user outage probability with $\alpha=2$ is given by substituting $z=\gamma_{\mathsf{th}}$ into~\eqref{eq: cdf of SINRd integral over r: alpha 2 Final}. Moreover, for  $\alpha=4$, the outage probability is lower bounded by substituting $z=\gamma_{\mathsf{th}}$ into~\eqref{eq: proof of the cdf of SINRd alpha final}.
\end{Corollary}
\begin{Corollary}
The DL user outage probability is given by substituting $z=\gamma_{\mathsf{th}}$ into~\eqref{eq: cdf SINRd MRC/MRT}.
\end{Corollary}

The $\MRCMRT$ scheme does not take into account the impact of the LI. Therefore, the system performance suffers under the impact of strong LI. Motivated by this, we now study the performance of more sophisticated linear combining schemes with superior LI suppression capability, namely the $\MRCZF$ and $\ZFMRT$ schemes.
\subsection{$\MRCZF$ Scheme}
In the $\MRCZF$ scheme the AP takes advantage of the available multiple transmit antennas to completely cancel the LI. To ensure this is possible, the number of the transmit antennas at AP should be greater than one, i.e., $\nt>1$.

With $\WR^{\MRC} = \frac{\vh_{\Us\AP}^{\dag}}{\parallel \vh_{\Us\AP} \parallel}$, the optimal combining vector $\WT$, which maximizes achievable UL and DL sum rate, should be the solution to the following maximization problem
\begin{align}\label{eq:sum rate max problem with ZF constraint}
\underset{\WT}{\text{max}}
&\quad \RFD =\log\left(1 + \frac{P_{\Us}\ell(x_{\Us} )}
{\Sn}\| \vh_{\Us\AP}\|^2\right)\nonumber\\
 &\quad\qquad+
\log\left(1 + \frac{P_{\AP} \ell(x_{\Ds}) \|\vh_{\AP\Ds}\WT\|^2}
{P_{\Us} g_{ \Us\Ds} \ell( x_{\Us},x_{\Ds}) + \Sn} \right)\nonumber\\
\text{s.t.}&
\quad\|\WR^{\MRC}\vH_{\AP\AP} \WT\|^2=0,~\|\WT\|^2 =1,
\end{align}
which can be further simplified as
\begin{align}\label{eq:sum rate max problem with ZF constraint simplified}
\underset{\WT}{\text{max}}
&\quad \|\vh_{\AP\Ds}\WT\|^2\\
\text{s.t.}&
\quad\|\WR^{\MRC}\vH_{\AP\AP} \WT\|^2=0,~\|\WT\|^2 =1.\nonumber
\end{align}
The solution is such that the vector $\WT$ is in the orthogonal complement space of $\WR^{\MRC} \vH_{\AP\AP}$.  The orthogonal projection onto the orthogonal complement of the column space of $\WR^{\MRC} \vH_{\AP\AP}$ is given by~\cite{Johnson:Book:1990}
\begin{align} \label{eq: null projection VR}
\Pi_{\vH_{\AP\AP}^{\dag}\WR^{\dag\MRC}}^{\perp} &= \vI_{\nt}-
\vH_{\AP\AP}^{\dag}\WR^{\dag\MRC}\nonumber\\
&\quad\times
\left(  \WR^{\MRC}\vH_{\AP\AP} \vH_{\AP\AP}^{\dag}\WR^{\dag\MRC}\right)^{-1}
\WR^{\MRC}\vH_{\AP\AP}.
\end{align}
Therefore, the optimal solution of~\eqref{eq:sum rate max problem with ZF constraint simplified} is given by
\begin{align} \label{eq: ZF wr}
\WT^{\ZF} = \frac{\Pi_{\WR \vH_{\AP\AP}}^{\perp} \vh_{\AP\Ds}^{\dag}}{\| \Pi_{\WR \vH_{\AP\AP}}^{\perp} \vh_{\AP\Ds}^{\dag}\|}.
\end{align}
Having obtained the ZF precoder with MRC processing, the received SINR at the AP and DL user can be obtained as
\begin{subequations}
\begin{align}
\SINRAP &=
\frac{P_{\Us}\ell(x_{\Us} )}
{\Sn}\| \vh_{\Us\AP}\|^2,\label{eq:SINRAP MRC/ZF}\\
\SINRd  &=
\frac{P_{\AP} \ell(x_{\Ds}) \|\Pi_{\WR \vH_{\AP\AP}}^{\perp} \vh_{\AP\Ds}^{\dag}\|^2}
{P_{\Us} g_{ \Us\Ds} \ell( x_{\Us},x_{\Ds}) + \Sn}.\label{eq:SINRd MRC/ZF}
\end{align}
\end{subequations}
Based on~\eqref{eq:achievable downlink rate}, in order to study the achievable sum rate of the UL and DL we investigate the cdf of $\SINRAP$ and $\SINRd $, respectively. Since $\|\vh_{\Us\AP}\|^2\sim\chi_{2\nr}^2$, we have
\begin{align}\label{eq:cdf of SINRa MRC/ZF}
F_{\SINRAP}(z)
 &= 1 - \frac{\lambda_{\Ds}}{\Gamma(\nr)}\int_{0}^{R_c}\int_{0}^{2\pi}
r e^{-\lambda_{\Ds}\pi r^2} \\
&\quad\times\Gamma\left(\nr, z\frac{ \Sn}{P_{\Us}}(r^2 + d^2-2rd\cos\theta)^{\frac{\alpha}{2}}\right)d\theta dr,\nonumber
\end{align}
which does not yield a closed-form solution, but can be evaluated numerically.

Next, we derive the cdf of $\SINRd$. For this purpose, let us first characterize $\|\Pi_{\WR \vH_{\AP\AP}}^{\perp} \vh_{\AP\Ds}^{\dag}\|^2= (\vh_{\AP\Ds}\Pi_{\WR \vH_{\AP\AP}}^{\perp} \Pi_{\WR\vH_{\AP\AP}}^{\dag\perp} \vh_{\AP\Ds}^{\dag} )$ which can be written as
\begin{align}\label{eq: numerator of the SINRa ZF}
\| \Pi_{\WR \vH_{\AP\AP}}^{\perp}\vh_{\AP\Ds}^{\dag}\|^2  &=
\left(\vh_{\AP\Ds}
\left(\vI_{\nt}-
\vH_{\AP\AP}^{\dag}\WR^{\dag\MRC}\right.\right.\nonumber\\
&\left.\left.\times\!
\left(  \WR^{\MRC}\vH_{\AP\AP} \vH_{\AP\AP}^{\dag}\WR^{\dag\MRC}\right)^{-1}\!\!
\WR^{\MRC}\vH_{\AP\AP}\right) \vh_{\AP\Ds}^{\dag}\!\right)\nonumber\\
&=
\left(\vh_{\AP\Ds} \boldsymbol{\Psi}_{t}
\left(\vI_{\nt}-\diag\{1,0,\cdots,0\}\right) \boldsymbol{\Psi}_{t}^{\dag}\vh_{\AP\Ds}^{\dag}\right)\nonumber\\
&=
\left(\hat{\vh}_{\AP\Ds}
\diag\{0,1,\cdots,1\} \hat{\vh}_{\AP\Us}^{\dag}\right),\nonumber\\
&= \| \tilde{\vh}_{\AP\Ds}\|^2,
\end{align}
where $\boldsymbol{\Psi}_{t}$ is an unitary matrix, $\hat{\vh}_{\AP\Ds}=\boldsymbol{\Psi}_{t}\vh_{\AP\Ds}$ and $\tilde{\vh}_{\AP\Ds}$ is a $(\nt-1)\times 1$ vector. Note that in~\eqref{eq: numerator of the SINRa ZF} the first equality holds because $\Pi_{\WR \vH_{\AP\AP}}^{\perp}$ is idempotent and the second equality is due to the eigen-decomposition. Hence, $\|\Pi_{\WR \vH_{\AP\AP}}^{\perp}\vh_{\Us\AP}\|^2 \sim \chi_{2(\nt-1)}^2$. By comparing~\eqref{eq:SINRd MRC/ZF} and~\eqref{eq:SINRd MRC/MRT} one can readily check that by replacing $\nt$ with $\nt-1$ in~\eqref{eq: cdf SINRd MRC/MRT}, the $\SINRd$ cdf, for the $\MRCZF$ scheme, is obtained.
\\\newline
\emph{\textbf{ Evaluation of $\boldsymbol{R_{\AP}}$:}}
By substituting~\eqref{eq:cdf of SINRa MRC/ZF} into~\eqref{eq:achievable downlink rate} the spatial average rate of the UL user can be written as
\begin{align}\label{eq:Ra MRC/ZF}
R_{\AP}
& = \frac{\lambda_{\Ds}}{\Gamma(\nr)}\int_{0}^{\infty}\int_{0}^{R_c}\int_{0}^{2\pi}
r e^{-\lambda_{\Ds}\pi r^2} \nonumber\\
&\times\Gamma\left(\nr, \epsilon_t\frac{ \Sn}{P_{\Us}}(r^2 + d^2-2rd\cos\theta)^{\frac{\alpha}{2}}\right) d\theta dr dt.
\end{align}
\\\newline
\emph{\textbf{ Evaluation of $\boldsymbol{R_{\Ds}}$:}}
Since the cdf of the received SINR at the DL user with the $\MRCZF$
scheme can be extracted from Proposition~\ref{Prop: Acheivable rate downlink}. Therefore, the average rate of the DL user under interference-limited case can be readily obtained by replacing $\nt$ with $\nt-1$ in~\eqref{eq: achivable sum rate DL user general} to yield
 \vspace{-0.3em}
\begin{align}\label{eq: achivable sum rate DL user MRC/ZF}
R_{\Ds}&=\!
 \frac{2}{\alpha\ln2}\!\left(\frac{R_c}{d}\right)^{\alpha}\frac{P_{\Us}}{P_{\AP} }
 \!\sum_{k=0}^{\infty}\frac{(-1)^k (\lambda_{\Ds}\pi R_c^2)^{k+1}}{\Gamma(k+1)}\\
&~\times
 G_{5 5}^{3 4} \left(\left(\frac{R_c}{d}\right)^{\alpha}\!\frac{P_{\Us}}{P_{\AP} } \
 \!\Big\vert \ \! {  0, \nt-1, 1-\frac{2(k+1)}{\alpha}, -1,0 \atop 1,-1,-1,\! -\frac{2(k+1)}{\alpha}, 1} \right).\nonumber
\end{align}

Remark~\ref{remark : new} indicates that DL transmission of $\MRCZF$ scheme exhibits an inferior performance compared to the $\MRCMRT$ scheme in terms of both outage probability and average rate of the DL user. This is intuitive, since the $\MRCMRT$ allocates one degree-of-freedom for LI cancellation at the transmit side of the AP.

 \vspace{-0.8em}
\subsection{$\ZFMRT$ Scheme}
As an alternative scheme we consider MRT with $\WT^{\MRT} =\frac{\vh_{\AP\Ds}^{\dag}}{ \parallel \vh_{\AP\Ds} \parallel}$ is used for transmit and optimize $\WR$ based on the ZF criterion. Note that to ensure that ZF can completely null the LI, the AP should equipped with $\nr>1$ receive antennas. In this case, the received SINR at the DL user is given by~\eqref{eq:SINRd MRC/MRT}. Furthermore, it is easy to show that $\SINRAP$ cdf can be obtained as~\eqref{eq:cdf of SINRa MRC/ZF} by replacing $\nr$ with $\nr-1$. The achievable UL and DL average sum rate can be easily derived as in the case of the $\MRCZF$ scheme and is omitted for the sake of brevity.
 \vspace{-1em}
\subsection{Optimal Solution With the Optimal Linear Receiver}
In this subsection, our main objective is to jointly design transmit and receive precoders so that the system achievable sum rate is
maximized. Specifically, the general sum rate maximization problem can be formulated as
 \vspace{-0.3em}
\begin{align}\label{eq:sum rate maximization problem formulation}
\underset{\WR,\WT}{\text{max}}
\quad \RFD &=\log_2\left(1 + a_1\|\vh_{\AP\Ds}\WT\|^2 \right) \nonumber\\
&+
\log_2\left(1 + \frac{ a_2\| \WR\vh_{\Us\AP}\|^2}
{P_{\AP} \|\WR\vH_{\AP\AP} \WT\|^2  + a_3}\right)\nonumber\\
\text{s.t.}&
\quad\|\WT\|^2 =1,~\|\WR\|^2 =1,
\end{align}
where $a_1 = \frac{P_{\AP} \ell(x_{\Ds})}{P_{\Us} g_{ \Us\Ds} \ell( x_{\Us},x_{\Ds}) + \Sn}$, $a_2 = P_{\Us}\ell(x_{\Us} )$, and $a_3 = \Sn$.

In order to solve the problem~\eqref{eq:sum rate maximization problem formulation}, we first fix $\WT$ and optimize $\WR$ to maximize the achievable sum rate. Note that given $\WT$, $\WR$ only influence the achievable UL rate.  Therefore, the optimization problem can be re-formulated as
 \vspace{-0.2em}
\begin{align}\label{eq:R_AP:optimal linear receiver}
\underset{\WR}{\text{max}}&
\quad R_{\AP}
=\log_2\left(1 + \frac{ a_2\WR\vh_{\Us\AP}\vh_{\Us\AP}^{\dag}\WR^{\dag}}
 {\WR(P_{\AP} \vH_{\AP\AP} \WT\WT^{\dag}\vH_{\AP\AP}^{\dag} + a_3\vI)\WR^{\dag}}\right)\nonumber\\
\text{s.t.}&
\quad\|\WR\|^2 =1.
\end{align}
Since logarithm is a monotonically  increasing function, we may equivalently solve,
\vspace{-0.3em}
\begin{align}\label{eq:R_AP:optimal linear receiver over SINR}
\underset{\WR}{\text{max}}&
\quad\frac{ \WR\vh_{\Us\AP}\vh_{\Us\AP}^{\dag}\WR^{\dag}}
 {\WR(P_{\AP} \vH_{\AP\AP} \WT\WT^{\dag}\vH_{\AP\AP}^{\dag} + a_3\vI)\WR^{\dag}}\\
\text{s.t.}&
\quad\|\WR\|^2 =1, \nonumber
\end{align}
which is a generalized Rayleigh ratio problem. It is well known that the objective function in~\eqref{eq:R_AP:optimal linear receiver} is globally maximized when $\WR$ is chosen as~\cite{Johnson:Book:1990}
\vspace{-0.3em}
\begin{align}\label{eq:ptimal linear receiver with fixed WT}
\WR  = \frac{\vh_{\Us\AP}^{\dag}(P_{\AP} \vH_{\AP\AP} \WT\WT^{\dag}\vH_{\AP\AP}^{\dag} + a_3\vI)^{-1}}
{\left\|\vh_{\Us\AP}^{\dag}(P_{\AP} \vH_{\AP\AP} \WT\WT^{\dag}\vH_{\AP\AP}^{\dag} + a_3\vI)^{-1}\right\|}.
\end{align}
Accordingly, the maximum UL achievable rate can be expressed as
\vspace{-0.3em}
\begin{align}\label{eq:R_AP:mat valu for optimal linear receiver}
&R_{\AP}^{\mathsf{max}}=\log_2\left(1 + a_2\vh_{\Us\AP}^{\dag}\left(P_{\AP} \vH_{\AP\AP} \WT\WT^{\dag}\vH_{\AP\AP}^{\dag} + a_3\vI\right)^{-1}\vh_{\Us\AP}\right),\nonumber\\
&=\log_2\left(\!1 \!+ \!a_2\vh_{\Us\AP}^{\dag}\left(\frac{\vI}{a_3} \!- \!\frac{(\sqrt{P_{\AP}} \vH_{\AP\AP} \WT)(\sqrt{P_{\AP}} \vH_{\AP\AP} \WT)^{\dag}}
{a_3^2 + a_3P_{\AP} \WT^{\dag}\vH_{\AP\AP}^{\dag}\vH_{\AP\AP} \WT }\right)\vh_{\Us\AP}\!\right)\!,
\nonumber\\
&=\log_2\left(1 \!+\! \frac{a_2}{a_3}\|\vh_{\Us\AP}\|^2 - \frac{a_2}{a_3}\frac{P_{\AP} \|\vh_{\Us\AP}^{\dag}\vH_{\AP\AP} \WT\|^2 }
{a_3 + P_{\AP} \WT^{\dag}\vH_{\AP\AP}^{\dag}\vH_{\AP\AP} \WT }\right)\!,
\end{align}
where the second equality is obtained by using the Sherman Morrison formula $(\vA + \vu\vv^{\dag})^{-1} = \vA^{-1} - (\vA^{-1}\vu\vv^{\dag}\vA^{-1})/(1 +\vv^{\dag}\vA^{-1}\vu)$ with $\vA=\vI$ and $\vu=\vv = \sqrt{P_{\AP}}\vH_{\AP\AP}\WT$. Therefore,  the optimization problem in~\eqref{eq:sum rate maximization problem formulation} is re-formulated as
\vspace{-0.0em}
\begin{align}\label{eq:sum rate maximization problem reformulation}
\underset{\WT}{\text{max}}&
\quad \RFD =
\log_2\left(1 + a_1\|\vh_{\AP\Ds}\WT\|^2 \right)\nonumber\\
& +
\log_2\left(1 + \frac{a_2}{a_3}\left(\|\vh_{\Us\AP}\|^2 - \frac{P_{\AP} \|\vh_{\Us\AP}^{\dag}\vH_{\AP\AP} \WT\|^2 }
{a_3 + P_{\AP} \WT^{\dag}\vH_{\AP\AP}^{\dag}\vH_{\AP\AP} \WT }\right)\right)\nonumber\\
\text{s.t.}&\quad \|\WT\|^2 =1,
\end{align}
which is still difficult to solve. Therefore, instead of solving it directly, we introduce an auxiliary
variable $t=\frac{P_{\AP} \|\vh_{\Us\AP}^{\dag}\vH_{\AP\AP} \WT\|^2 }
{a_3 + P_{\AP} \WT^{\dag}\vH_{\AP\AP}^{\dag}\vH_{\AP\AP} \WT } $. Then, assuming $t$ is known, we solve the optimization~\eqref{eq:sum rate maximization problem reformulation} as
\vspace{-0.3em}
\begin{align}\label{eq:sum rate maximization with parameter t}
\underset{\WT}{\text{max}}&
\quad\|\vh_{\AP\Ds}\WT\|^2\\
\text{s.t.}&
\quad\frac{P_{\AP} \|\vh_{\Us\AP}^{\dag}\vH_{\AP\AP} \WT\|^2 }
{a_3 + P_{\AP} \WT^{\dag}\vH_{\AP\AP}^{\dag}\vH_{\AP\AP} \WT } =t,\quad\|\WT\|^2 =1.\nonumber
\end{align}
This is a nonconvex quadratic optimization problem with quadratic equality constraint and difficult to solve. To solve the problem in~\eqref{eq:sum rate maximization with parameter t}, we apply a similar approach as in~\cite{Zheng:JSPL2013} to convert the optimization problem~\eqref{eq:sum rate maximization with parameter t} to
\begin{align}\label{eq:sum rate maximization SDP}
\underset{\WT}{\text{max}}&
\quad{\mathsf{trace}}(\vh_{\AP\Ds}\tilde{\vW}_{t}\vh_{\AP\Ds}^{\dag})\\\
\text{s.t.}&
\quad{\mathsf{trace}} (\tilde{\vW}_{t}(\vH_{\AP\AP}^{\dag} \vh_{\Us\AP}\vh_{\Us\AP}^{\dag}\vH_{\AP\AP} - t\vH_{\AP\AP}^{\dag}\vH_{\AP\AP})) = \frac{a_3}{P_{\AP}}t,\nonumber\\
&\quad
\tilde{\vW}_{t} \succeq 0, \quad{\mathsf{trace}}(\tilde{\vW}_{t}) =1, \quad\text{rank}(\tilde{\vW}_{t})=1,\nonumber
\end{align}
where $\tilde{\vW}_{t} = \vw_t\vw_t^{\dag}$  is a symmetric, positive semi-definite (PSD) matrix. Note that~\eqref{eq:sum rate maximization SDP} is still nonconvex due to the rank-1 constraint. But we can resort to widely used semidefinite relaxation (SDR) technique to solve it. In SDR, the rank-1 constraint is first dropped and the resulting problem becomes a semidefinite programming (SDP), whose solution $\tilde{\vW}_{t}^{\dag}$ can be found by using the method provided in~\cite[Appendix B]{Zheng:JSPL2013} or by using appropriate solvers e.g., Gurobi. Once $\tilde{\vW}_{t}^{\dag}$ is obtained, we can check if it satisfies the rank-1 constraint. Note that, in~\cite[Appendix B]{Zheng:JSPL2013}, it has been shown that the optimal solution is rank-1 and hence~\eqref{eq:sum rate maximization SDP} and its SDR are equivalent. Therefore, the optimal $\vw_t^{\dag}$ of~\eqref{eq:sum rate maximization with parameter t} can be extracted from $\tilde{\bm{W}}_t^\star$. Denoting the optimal objective value of~\eqref{eq:sum rate maximization SDP} as $h(t)$, the achievable sum rate maximization problem can be formulated as
\begin{align}\label{eq:sum rate maximization problem one dimentional}
\underset{t\geq 0}{\text{max}}&
\quad \RFD (t) =
\log_2\left(\left(1\! +\! a_1 h(t)\right) \left(1 + \frac{a_2}{a_3}\left(\|\vh_{\Us\AP}\|^2 \!- \!t\right)\right)\right).
\end{align}
Therefore, in order to solve~\eqref{eq:sum rate maximization problem formulation}, it remains to perform a one-dimensional optimization with respect to the variable $t$.
\subsection{Comparison of the Proposed Schemes}
Here, we provide a brief discussion on the implementation complexity of the four proposed schemes, in terms of the amount of channel state information (CSI) required for their operation and computational complexity. In practice, the acquisition of CSI involves additional signaling overhead for channel estimation, which must be considered in the design of wireless systems. On the other hand, if a large amount of CSI is available at the transmitting node, more sophisticated transmission schemes could be designed to improve the transmission efficiency and to achieve a better performance. The $\MRCMRT$ scheme has the lowest CSI requirement of the four, since it only needs the CSI knowledge of the UL and DL channels. On the other hand, the remaining three schemes additionally require the CSI knowledge of the LI channel. The computational complexity of the optimal solution is much higher than the other three methods as it involves inversion of the (high-dimension) matrices and solving a SDP problem. Since one-dimensional optimization along $t$ is required,~\eqref{eq:sum rate maximization SDP} needs to be solved $N$ times, where $N$ is the number of quantization point on $t$ if an exhaustive search along $t$ is performed. Then with each $t$,  solving~\eqref{eq:sum rate maximization SDP} requires running time of $O(\nt^{4.5})$~\cite{Zhi-quan:SDP:2010}, where $\nt$ is the length of $\WT$. Therefore, the total running time is $O(N\nt^{4.5})$.
\section{Baselines For Comparison}
In this section, in order to obtain more insight, and support our average sum rate results, we derive a tight upper bound for the UL and DL average sum rate in case of $\nt=\nt=1$. Moreover, we consider the sum rate of the $\MRCMRT$ scheme with large receive antennas due to the recent interest on massive MIMO technology with low complexity linear combining. We also investigate the average sum rates due to the HD mode of operation as a baseline reference for comparison with the FD counterpart.
\subsection{Dual-Antenna AP}
Consider the special case of $\nr=\nt=1$. In order to facilitate a closed-form analysis, we neglect the effect of LI at the AP and internode interference at the DL user.  The cdf of the $\SNRa$ for $\alpha=2$ can be derived as
\vspace{-0.1em}
\begin{align} \label{eqn:cdf_SNRa_Asyp}
F_{\SNRa}(z) &=1-\left(1+\frac{z}{\psi_{\Us}}\right)^{-1}e^{-\frac{\lambda_{\Ds}\pi d^2}{1 + \frac{\psi_{\Us}}{z}}},
\end{align}
where $\psi_{\Us}=\frac{P_{\Us}}{\Sn}\lambda_{\Ds}\pi$. Therefore, by using~\eqref{eq:achievable downlink rate}, an upper bound on the average rate of the UL user is
\vspace{-0.3em}
\begin{align} \label{eqn:achievabe rate UP}
R_{\AP} &\leq\frac{1}{\ln2}
\left(\frac{1}{\psi_{\Us}}-1\right)^{-1}e^{-\frac{\lambda_{\Ds}\pi d^2 }{\psi_{\Us}}}\nonumber\\
&~\times\left(
\Ei\left(\frac{\lambda_{\Ds}\pi d^2}{1-\psi_{\Us}}\right)-
\Ei\left(\frac{\psi_{\Us}\lambda_{\Ds}\pi d^2}{1-\psi_{\Us}}\right)\right).
\end{align}

Similarly, by neglecting the term $P_{\Us} g_{ \Us\Ds} \ell( x_{\Us},x_{\Ds})$ in~\eqref{eq:SINR: downlonk user}, a valid assumption for $P_{\Us} d^{-\alpha}\ll1$, we  obtain
\vspace{-0.3em}
\begin{align} \label{eqn:cdf_SNRd_Asyp}
 F_{\SNRd}(z) &=
\left\{%
 \begin{array}{clcr}
  1-\left(1+\frac{z\lambda_{\Ds}\pi}{\psi_{\Ds}}\right)^{-1}&                  \alpha=2,  \\
  1- \sqrt{\frac{\psi_{\Ds}}{2z}} e^{\frac{\psi_{\Ds}}{8z}}D_{-1} \left( \sqrt{\frac{\psi_{\Ds}}{2z}}\right)                 &     \alpha=4,
  \end{array}%
\right.
\end{align}
where $\psi_{\Ds} = \frac{P_{\AP}}{\Sn}(\lambda_{\Ds}\pi)^2 $. Hence, the UL user average rate is upper bounded as
\vspace{-0.2em}
\begin{align} \label{eqn:achievabe rate DL}
R_{\Ds} &\leq\frac{1}{\ln2}
\left\{%
 \begin{array}{clcr}
  \left(\frac{\lambda_{\Ds}\pi}{\psi_{\Ds}}-1\right)^{-1}\log\left(\frac{\lambda_{\Ds}\pi}{\psi_{\Ds}}\right)&                  \alpha=2,  \\
  \int_{0}^{\infty}\frac{1}{z+1}\sqrt{\frac{\psi_{\Ds}}{2z}} e^{\frac{\psi_{\Ds}}{8z}}D_{-1} \left( \sqrt{\frac{\psi_{\Ds}}{2z}}\right)dz              &     \alpha=4.
  \end{array}%
\right.
\end{align}

It is worthwhile to point out that the corresponding outage probability of the UL and DL transmission with $\nr=\nt=1$ can be determined by substituting $z=\gamma_{th}$ into~\eqref{eqn:cdf_SNRa_Asyp} and~\eqref{eqn:cdf_SNRd_Asyp}, respectively.
\subsection{$\MRCMRT$ Scheme With a Large Receive Antenna Array}
In the $\MRCMRT$ scheme LI plays a major role to limit the performance. Motivated by this and the developments in the area of massive MIMO~\cite{Quoc:2013:COM}, we now consider the case of large receive array as a simple way of removing the effect of LI~\cite{Ngo:JSAC:2014}. It is interesting to observe that for any finite $\nt$, as $\nr$ grows large, the channel vectors of the desired signal and the LI become nearly orthogonal. Therefore, the MRC receiver can act as an orthogonal projection of the LI\cite{Ngo:JSAC:2014}. Note that from~\eqref{eq:SINR at AP in uplink}, the received SINR at the AP can be written as
\begin{align}\label{eq:SINR at AP in uplink for infinite nr}
\SINRAP =
\frac{P_{\Us}   \ell(x_{\Us} ) \| \WR\vh_{\Us\AP}\|^2}
{P_{\AP} \sum_{i=1}^{\nt}\|\WR^{\MRC}\vh_{\AP i}\|^2w_{ti} + \Sn\|\WR\|^2}.
\end{align}
where $\vh_{\AP i}$ is the $i$th column of $\vH_{\AP\AP}$ (i.e., $\vH_{\AP\AP} = [\vh_{\AP1},\vh_{\AP2},\cdots,\vh_{\AP\nt}]$) and $w_{ti}$ is the $i$th element of $\WT$. Then from the law of large numbers for the asymptotic large $\nr$ regime, we have
\vspace{-0.2em}
\begin{align}\label{eq:the law of large number infinite nr}
\frac{1}{\nr} \vh_{\Us\AP}^{\dag}\vh_{\AP i}\xrightarrow{a.s.}0,~~ \text{as}~\nr\rightarrow \infty,
\end{align}
where $\xrightarrow{a.s.}$ denotes the almost sure convergence. As a result, the LI can be reduced significantly by scaling the AP transmit power with $\nr$ together with the MRC receiver. Hence, the average sum rate for the $\MRCMRT$ scheme with asymptotic large $\nr$ regime can be expressed as
\begin{align}\label{eq: sum rate of FD AP for larg n regime}
\RFL =&
\E{\log_2\left(1+\frac{P_{\AP} \ell(x_{\Ds}) \|\vh_{\AP\Ds}\|^2 }
{\nr(P_{\Us} g_{ \Us\Ds} \ell( x_{\Us},x_{\Ds}) + \Sn)}\right)}\nonumber\\
&~
+\E{\log_2\left(1+\frac{P_{\Us}   \ell(x_{\Us} ) }{ \Sn}\|\vh_{\Us\AP}\|^2\right)},
\end{align}
where~\eqref{eq: achivable sum rate DL user general} (after replacing $P_{\AP}$ with $P_{\AP}/\nr$) provides an expression for the first expectation term in the interference-limited case. Moreover, the right hand side expectation term is given by~\eqref{eq:Ra MRC/ZF}.
\subsection{Half-Duplex Mode}
In this subsection, we compare the performance of the HD and FD modes of operation at the AP. In the HD mode of operation,  AP employs orthogonal time slots to serve the UL and DL user, respectively. In order to keep our comparisons fair, we consider \emph{``antenna conserved''} (AC) and \emph{``RF-chain conserved''} (RC)
scenarios which are adopted in the existing literature~\cite{Khojastepour:Mobicom:2012}. Under AC condition, the total number of antennas used by the HD AP and FD AP are kept identical. However, the number of RF chains employed by the HD AP is higher than that of the FD AP~\cite{Khojastepour:Mobicom:2012} and hence former system would be a costly option. Under RC condition, the total number RF chains used HD and FD modes are kept identical. Therefore, in DL (or UL) transmission, the HD AP only uses $\nt$ (or $\nr$) antennas under the RC condition, while it uses $\nt+\nr$ antennas under the AC condition.

The average sum rate under the RC condition, using the weight vector $\vw_r^{\MRC} = \frac{\vh_{\Us\AP}^{\dag}}{\|\vh_{\Us\AP}\|}$ for the MRC receiver, and the MRT precoding vector $\vw_t^{\MRT} = \frac{\vh_{\AP\Ds}^{\dag}} {\|\vh_{\AP\Ds}\|}$ can be expressed as
\begin{align}\label{eq: sum rate of single-antenna HD AP}
\RHDs&=\!\delta\E{\log_2\left(1+\!\snr_{\Ds,\RC}\ell(x_{\Ds})\| \vh_{\AP\Ds}\|^2\right)}\nonumber\\
&~+\!(1-\!\delta)\E{\log_2\left(1+\!\snr_{\Us,\RC}\ell(x_{\Us})\| \vh_{\Us\AP}\|^2\right)},
\end{align}
where $\delta$ ($0<\delta<1$) is a fraction  of the time slot duration of $T$, used for DL transmission, $\snr_{\Ds,\RC} = \PaHR/\Sn$, and $\snr_{\Us,\RC} = \PuHR/\Sn$, where $\PaHR$ and $\PuHR$ are the transmit power of the AP and UL user, respectively, in the HD-RC mode.

Under the AC condition, the average achievable rate can be expressed as
\begin{align}\label{eq: sum rate of dual-antenna HD AP}
\RHDd =& \delta\E{\log_2\left(1+\snr_{\Ds,\AC}\ell(x_{\Ds})\|\vh_{\AP\Ds}\|^2\right)}\\
&~+ (1-\delta)\E{\log_2\left(1+\!\snr_{\Us,\AC} \ell(x_{\Us})\| \vh_{\Us\AP}\|^2\right)},\nonumber
\end{align}
where $\snr_{\Ds,\AC} = \PaHA/\Sn$, and $\snr_{\Us,\AC} = \PuHA/\Sn$, where $\PaHA$ and $\PuHA$ are the transmit power at the AP and UL user, respectively.

Using~\eqref{eq:Ra MRC/ZF} with change of variables, the second expectation of~\eqref{eq: sum rate of single-antenna HD AP} and~\eqref{eq: sum rate of dual-antenna HD AP} can be obtained. Moreover, after some algebraic derivations, we get
\begin{align}\label{eq:average capacity of UL user MRC/ZF}
&\E{\log_2\left(1+\snr_{\Ds,i}\ell(x_{\Ds})\| \vh_{\AP\Ds}\|^2\right)}  =
\frac{2}{\alpha\ln2}\frac{1}{\Gamma(n_t)}\times\nonumber\\
&\sum_{k=0}^{\infty}\frac{(-1)^k(\lambda_{\Ds}\pi R_c^2)^{k+1}}{\Gamma(k+1)}
G_{3 4}^{3 2} \left( \frac{\Sn}{P_{\AP}^{\HD-i}}  R_c^{\alpha}  \
 \Big\vert \  {1-\frac{2(k+1)}{\alpha}, 0,  1 \atop 0, n_t, 0, -\frac{2(k+1)}{\alpha}} \right),
\end{align}
where $i\in\{\RC,\AC\}$ and under RC and AC conditions $n_t=\nt$ and $n_t=\nt+\nr$, respectively.
\begin{figure}[t]
\centering
\vspace{-1.1em}
\includegraphics[width=91mm, height=68mm]{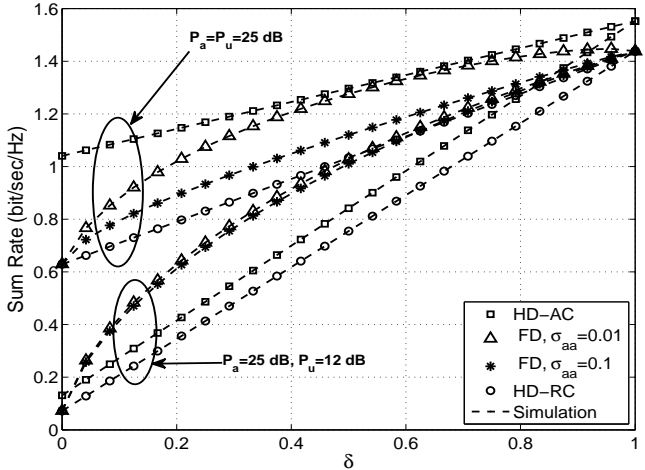}
\vspace{-1.3em}
\caption{UL and DL average sum rate versus $\delta$ for the FD and HD AP  ($\nt=\nr = 1$, and $d=25$ m). Simulation results are shown by dashed lines.}
\vspace{-1.0em}
\label{fig: sum_rate_v_delta}
\end{figure}

We end this section with the following remarks. In general, the corresponding SNRs for DL and UL transmissions in HD mode are larger than those of in FD mode. However, although HD mode does not induce LI and internode interference, it imposes a pre-log factor $\delta$ on the spectral efficiency. Since, most of the results contain Meijer G-functions, a direct comparison of the average sum rate of the FD and HD modes is challenging. Nevertheless, let us consider the achievable rate region of both FD and HD modes. The rate region frontiers can be found by sweeping $\delta$ over the full range of $[0,~1]$, while total energy of the AP and UL user for FD and HD modes are the same. From~\eqref{eq: sum rate of single-antenna HD AP} and~\eqref{eq: sum rate of dual-antenna HD AP} we observe that the achievable rate region of the HD mode is a linear decreasing function of $\delta$. On the other hand, by modeling the FD network as a two-user interference channel as in~\cite{Paulraj:2007}, frontier of the achievable rate region of $R_{\AP}$ and $R_{\Ds}$ (defined in~\eqref{eq:achievable rate FD}) through the power levels of $\delta P_{\AP}$ and $(1-\delta)P_{\Us}$, are convex (concave in case of $\MRCMRT$ scheme with high LI strength)~\cite{Paulraj:2007} with the same extremity points in the frontiers as HD system. This observation indicates that with the ideal choice of the linear processing scheme as well as the time fraction, the potential gains of the FD mode over the HD mode can be always exploited.

\section{Numerical Results and Discussion}\label{sec:Numerical results}
In this section, we evaluate the system performance and elucidate the effect of system parameters on achievable UL and DL average sum rate. Specifically, we compare the performance of different precoding schemes investigated in Section~\ref{sec:Performance Analysis} in terms of the achievable sum rate. Unless otherwise stated, the value of network parameters are: $\alpha=2$, $R_c=200$ m, and $\lambda_{\Ds}=1\times 10^{-3}$ node/$\text{m}^2$. The state-of-the-art work has demonstrated that LI can be significantly suppressed via a combination of various analog and digital techniques~\cite{Duarte:PhD:dis}. As such, the typical values of $0.01$ ($-20$ dB) and $0.1$ ($-10$ dB) for $\Sap$, are used in the simulations. Moreover, with curves shown in Figs. 5-9, we assume that the total power of the AP and UL user for FD and HD modes are the same. The curves in Figs. 2-9 were plotted using the developed analytical expressions in Section~\ref{sec:Performance Analysis} except in the cases for $\MRCZF$ (with $\nt,\nr>1$) and optimal where we have used simulations.

\begin{figure}[t]
\vspace{-1em}
\includegraphics[width=91mm, height=68mm]{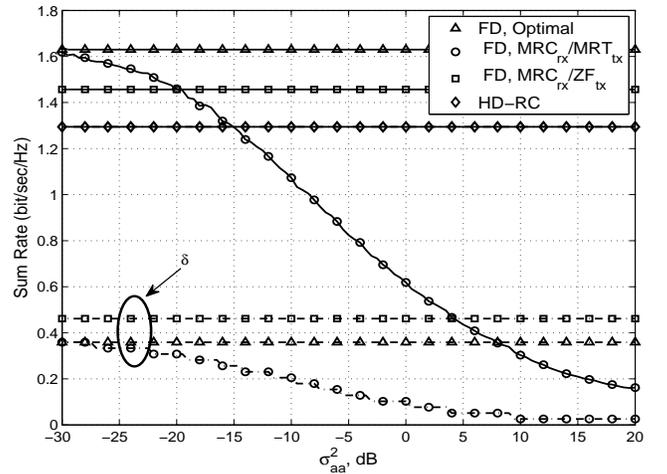}
\vspace{-1.4em}
\caption{Average sum rate versus $\Sap$. $\delta$ is adjusted to support the same amount of traffic for both DL and UL users ($\nt=\nr=3$).}
\vspace{-1.4em}
\label{fig: sum_rate_v_sap fair}
\end{figure}

\subsection{Effect of Resource Allocation}
In Fig.~\ref{fig: sum_rate_v_delta} we compare the UL and DL average sum rate as a function of $\delta$ for the FD and HD modes of operation and for two different values of $\sigma_{\AP\AP}$. Here, $\nt=\nr=1$ is used for FD operation and also the $\MRCMRT$ scheme is studied for the HD-AC mode. The total energy consumption for both the FD and HD modes is the same and we plot the UL and DL average sum rate for two different power constraints $(P_{\AP},P_{\Us}) =(25~\text{dB}, 25~\text{dB} )$ (symmetric) and $(P_{\AP},P_{\Us}) =(25~\text{dB}, 12~\text{dB})$ (asymmetric), respectively. In particular, numerical results lead to the following conclusions: \emph{$1$)} As expected, the UL and DL average sum rate under the RC condition is worse than those of other cases. \emph{$2$)} In the asymmetric case, FD operation outperforms HD within the practical range of $\delta$. However, in the symmetric case, HD-AC condition achieves the best performance even for $\sigma_{\AP\AP}=0.01$. \emph{$3$)} It is clear that the symmetric case is more vulnerable to the strength of the LI.

In Fig.~\ref{fig: sum_rate_v_sap fair}, we plot the UL and DL average sum rate of the FD and conventional
HD system versus self LI channel gain $\Sap$ for $\nt=\nr=3$ and $d=25$ m. In order to guarantee the fairness of DL and UL users, $\delta$ is numerically adjusted to support the same amount of traffic for both DL and UL users, while the average sum rate is maximized\footnote{Please note that, in order to strike a balance between maximizing the system average sum rate and maintaining fairness among users, the idea of proportional fair scheduling can also be applied to our framework. We postpone this problem to our future work.}. Fig.~\ref{fig: sum_rate_v_sap fair} depicts that with appropriate choice of $\delta$, the FD mode outperforms its HD counterpart as long as the LI is sufficiently canceled. Note that in case of $\MRCMRT$, $\delta$ is decreased as $\Sap$ is increased, while in other cases $\delta$ is constant. This is intuitive, since the uplink average rate is degraded when the LI strength is intensified. Therefore, to guarantee the fairness $\delta$ must be decreased which lowers the residual LI power at AP (i.e., $\delta P_{\AP} \|\WR\vH_{\AP\AP} \WT\|^2$) and consequently increases the SINR.

\begin{figure}[t]
\centering
\includegraphics[width=90mm, height=64mm]{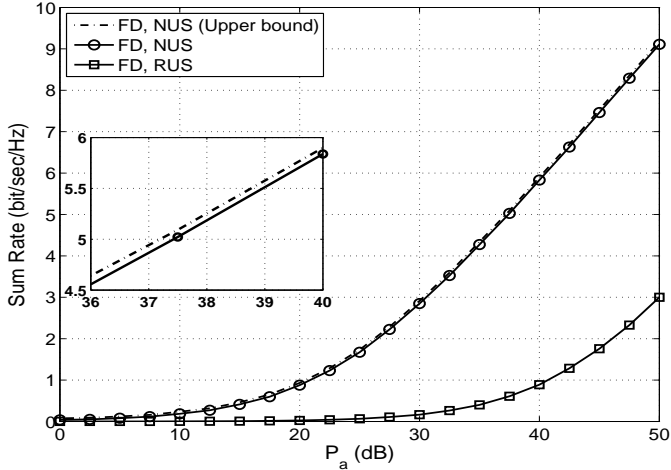}
\vspace{-1.5em} \caption{UL and DL average sum rate versus $P_{\AP}$ for nearest user selection (NUS) and random user selection (RUS) ($\nt=\nr=1$, $d=25$ m, $P_{\Us}=10$ dB, and $\Sap=0.1$).}
\vspace{-1.4em}
\label{fig: sum_rate_v_P_NUS_RUS}
\end{figure}

Fig.~\ref{fig: sum_rate_v_P_NUS_RUS} shows the UL and DL average sum rate versus $P_{\AP}$ for the nearest DL user (to the AP) and UL user for $\nt=\nr=1$, $\Sap= 0.1$ and $d=25$ m. The UL and DL average sum rate of the RUS scheme is also included as a benchmark comparison. Moreover, the upper bound is plotted using~\eqref{eqn:achievabe rate UP} and~\eqref{eqn:achievabe rate DL} with $P_{\Us} =10$ dB for the nearest UL user and the DL user, respectively. We see that the analytical bound provides a very tight bound on the UL and DL average sum rate. Furthermore, as expected, we see that the nearest user selection (NUS) scheme outperforms the RUS scheme.
\begin{figure}[t]
\centering
\vspace{-1.1em}
\includegraphics[width=90mm, height=67mm]{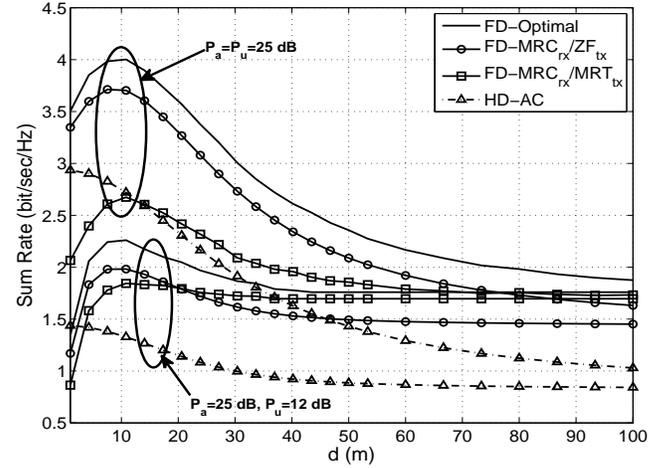}
\vspace{-1.2em} \caption{ UL and DL average sum rate versus $d$ for FD AP and HD AP. ($\nt=\nr=5$, $\Sap=0.01$, and $\delta=0.5$).}
\vspace{-1em}
\label{fig: sum_rate_v_d}
\end{figure}

\subsection{Effect of Internode Interference and LI}
We investigate the impact of internode interference and LI on the average sum rate of the system taking into account two scenarios based on $d$ and $\Sap$ as follows: First, the internode distance is increased, while the LI power level is kept constant at $-20$ dB ($\Sap=0.01$). Fig.~\ref{fig: sum_rate_v_d} plots the average sum rate against distance $d$ between the UL and DL user achieved by the FD and HD modes of operation and for different precoding schemes. Here, $\nt=\nr=5$ antennas are deployed at AP and two different power constraint are considered. Results not shown here for HD-RC showed an inferior performance as compared to the HD-AC case.

\begin{itemize}
\item It can be seen that for large $d$, the average sum rate of the FD mode is higher than that of the HD mode. The effect of large $d$ on the sum rate is visible as the internode interference effect becomes weak a performance gain of up to twice (for $\Sap=0$) that of the HD mode can be obtained.

\item Observe that there clearly is an optimal location for UL user placement around its corresponding DL user for maximizing the sum rate of each precoding design at FD mode of operation. At this optimal point, both DL and UL transmission are fairly supported by the AP, whereas the internode interference is not overwhelming. Nevertheless, in the case of the $\MRCMRT$ scheme, under the symmetric power constraint the HD mode outperforms the FD mode at this optimum location.

\item
We observe that FD transmission with the $\MRCZF$ scheme can significantly increase the achievable UL and DL sum rate. However, as $d$ increases, the curves for the FD mode intersect and the $\MRCMRT$ scheme surpasses the $\MRCZF$ scheme. This trend can be justified by noting that the effect of DL transmission for this simulation setup is more dominant than the UL transmission. We recall that when $d$ is increased or $P_{\Us}$ is decreased (i.e., the asymmetric power constraint case) $\SINRAP$ is degraded for both the $\MRCMRT$ and $\MRCZF$ precoding schemes. On the other hand, since the MRT precoder is optimal in sense of maximizing the $\SINRd$, it is not surprising that MRT precoder attains a better performance than that of the ZF.
\end{itemize}
\begin{figure}[t]
\centering
\includegraphics[width=90mm, height=65mm]{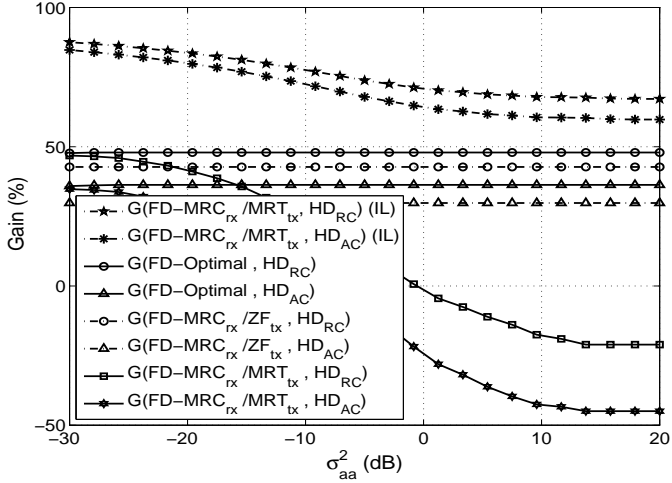}
\vspace{-1.6em}
 \caption{Average sum rate gain of the system ($\nt=\nr=5$, $d=25$ m, $P_{\Us}=P_{\AP}=25$ dB, and $\delta=0.5$).}
\vspace{-1.3em}
\label{fig: Average_spectral_efficiency_gain}
\end{figure}
Next, we fix the distance between the UL and DL users to be $d=25$ m and let the LI power increase from $-30$ dB to $20$ dB.
In Fig.~\ref{fig: Average_spectral_efficiency_gain} we plot the average sum rate gain, which is defined as $G(\FD,\HD_{i})\! = (R_{\mathsf{sum}}^{\FD}-R_{\mathsf{sum}}^{\HD-i})/R_{\mathsf{sum}}^{\FD}$ versus $\Sap$ and for $\nt=\nr=5$ and $\delta=0.5$. The average sum rate gain of the interference-limited FD mode with the $\MRCMRT$ scheme is also included for comparison (dashed line curves). A general observation is that the FD mode with the $\MRCZF$ scheme significantly and consistently outperforms the HD counterpart in all regimes of LI. Nevertheless, when the LI is low ($\Sap < -22$ dB) sum rate gain achieved by the $\MRCZF$ scheme appears to be limited when compared with the $\MRCMRT$ scheme. This is due to the fact that, when LI is substantially suppressed, the same SINR yields at the AP for both schemes.
As such, both schemes achieve a similar average rate in the UL channel. On the other hand, since MRT processing is the optimum precoder for SNR maximization in absence of the interference, the received SINR at the scheduled DL user (and consequently the DL user average rate) is slightly higher than that of the ZF precoder. Moreover, we observe that when $\Sap$ increases, the average sum rate gain of the $\MRCMRT$ scheme decreases. The sum rate gain loss is more pronounced, when $\Sap\geq -15$ dB, as for $\Sap\geq -10$ dB the HD-AC mode outperforms the FD mode (for asymmetric power case, please see Fig.~\ref{fig: sum_rate_v_snr_EqPower}.). Now, comparing the optimal scheme and other precoding schemes, we see that, as expected the optimal design can achieve respectively, up to $37\%$ and $47\%$ average sum rate gains in comparison with HD-AC and HD-RC schemes in all LI regimes. Note that depending on transmit powers and LI strength, the $\MRCMRT$ scheme is preferred over the $\MRCZF$ design and vice versa. The critical factor is the LI level. If all other factors are fixed, the $\MRCZF$ scheme always outperforms the $\MRCMRT$ scheme if $\Sap >  \sigma^2_{\AP\AP,0}$, where $\sigma^2_{\AP\AP,0}$ is the root of $\RFMM=\RFMZ$.
\begin{figure}[t]
\centering
\includegraphics[width=90mm, height=65mm]{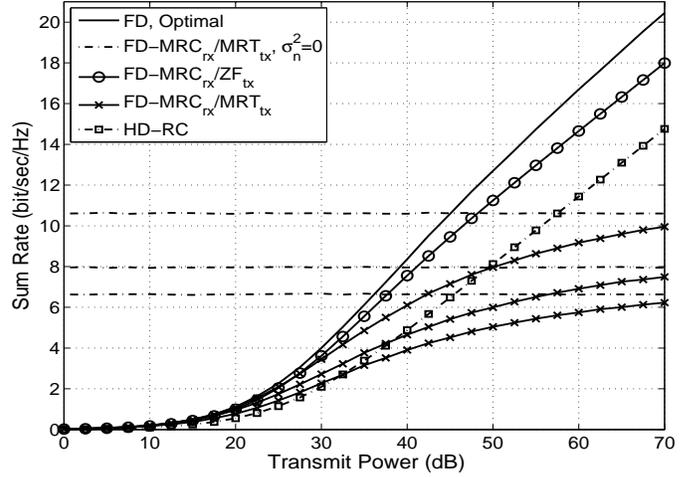}
\vspace{-1.4em}
 \caption{UL and DL average sum rate versus transmit power for the FD  and HD AP and for different LI strengths ($\nt=\nr=5$, $d=25$ m, $P_{\Us}=0.5P_{\AP}$, and $\delta=0.5$).}
\vspace{-1.7em}
\label{fig: sum_rate_v_snr_EqPower}
\end{figure}

Fig.~\ref{fig: sum_rate_v_snr_EqPower} shows the UL and DL average sum rate versus the transmit power, where the AP transmit power level is twice that of the UL user, i.e., we consider asymmetric power case. We compare the FD and HD modes for different levels of LI. Here $\nt=\nr=5$ antennas are deployed at the AP and $d$ is set to be $25$ m. We also plot the curves for the UL and DL average sum rate of the interference-limited FD mode with the $\MRCMRT$ scheme, as a performance upper bound. As expected, it can be seen that this upper bound is inversely proportional to the average strength of the LI channel.
Moreover, when transmit power is increased, the $\MRCMRT$ sum rate saturates to the corresponding upper bound, and reconfirms the correctness of our analysis. This observation simply means that the LI cancellation mechanism should be highly effective in the FD mode to compete against the HD.
\subsection{Effect of Antenna Number and Configuration}
We consider the influence of transmit/receive antenna number and also the antenna configuration at the FD AP on the achievable UL and DL sum rate. Fig.~\ref{fig: sum_rate_v_nu} compares the transmission rates for FD and HD AP in the cases of fixed internode interference and LI power level as a function of the number of transmit/receive antenna at the AP. As expected for severe LI conditions, since ZF precoder can completely cancel the LI, the  relative performance gap between the $\MRCZF$ scheme and the other schemes is notable. In particular,  the relative gap between the FD and HD curves  gradually increases when the number of antennas increases. This observation reveals that deployment of large number of antennas is beneficial to the performance. The UL and DL average sum rate of the optimal precoder design is also plotted for completeness of comparison. It can be seen that optimal precoder yields the best performance among all schemes.

Fig.~\ref{fig:sum_rate_Antenna Config} shows the UL and DL average sum rate versus transmit power for the $\MRCMRT$ scheme for an asymmetric power allocation case, i.e., $P_{\AP}=2P_{\Us}$. We compare the FD and HD modes for four different transmit/receive antenna pairs: $[\nt, \nr]\in\{[1,1], [2, 1], [1,2], [2, 2]\}$. We also included the benchmark performance achieved by interference-limited assumption for the first three antenna pairs. This results implies that the achievable rate given by Proposition~\ref{Propos:average capacity of the UL user MRC/MRT}, Proposition~\ref{Prop: Acheivable rate downlink}, and~\eqref{eq:SINRa nt=1, cdf pdf interference limited} are good predictors of the system's sum rate.
\begin{figure}[t]
\centering
\includegraphics[width=90mm, height=67mm]{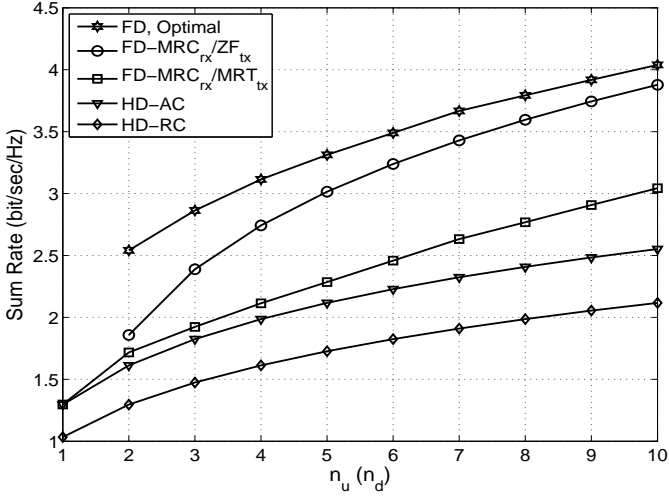}
\vspace{-1.2em}
\caption{UL and DL average sum rate versus $\nt=\nr$ for the FD and HD AP and for different precoding schemes ($d=25$ m, $P_{\AP}=P_{\Us}=25$ dB, $\Sap=0.1$, and $\delta=0.5$).}
\vspace{-0.7em}
\label{fig: sum_rate_v_nu}
\end{figure}
An interesting observation that can be extracted from Fig.~\ref{fig:sum_rate_Antenna Config} is that, increasing the receive antenna number at the FD AP is more beneficial to the UL and DL average sum rate than increasing the number of transmit antenna elements. For example, the relative performance gain of $[1,2]$ scheme over $[1,1]$ is more than that of $[2,1]$ scheme over $[1,1]$, especially at high transmit power levels of $>25$ dB. As noted before, this can be explained by the fact that by doubling the number of transmit/receive antenna number at the AP, the numerator of $\SINRAP$~\eqref{eq:SINRAP MRC/MRT} and $\SINRd$~\eqref{eq:SINRd MRC/MRT} is increased in the same proportions. However, by doubling the $\nt$, the LI at AP is boosted, leading to a decrease in $\SINRAP$ and consequently in the average sum rate.

\section{Conclusion}\label{sec:conclusion}
The  performance of a wireless network scenario in which a multiple antenna equipped FD AP communicates with spatially random single-antenna user nodes in the UL and DL channels simultaneously has been analyzed.  In particular, we considered precoding schemes based on the principles of MRC, MRT and ZF and studied the system performance in terms of the UL and DL average sum rate. Further, we have considered the problem of optimal precoding design for the UL and DL sum rate maximization and reformulated the problem as a SDP, which can be efficiently solved. Analysis and simulation results demonstrated the superiority of the optimal precoding scheme over the MRC/MRT and MRC/ZF schemes. We further studied the effect of resource allocation, LI and internode interference, and antenna configuration on the system sum rate. We found that the MRC/MRT scheme can offer a higher UL and DL average sum rate compared to the MRC/ZF scheme, when the LI is significantly canceled or the internode interference is weak enough, and vice versa. Furthermore, we observed that the performance gap between the FD and HD modes can be further increased by deploying more transmit/receive antennas at the AP. As for future research work, the performance gains due to FD transmission in setups such as heterogeneous network architectures with mixed FD/HD mode operation, cooperative relaying and MIMO may be characterized to further establish the viability of the usage of FD terminals.
\begin{figure}[th]
\centering
\includegraphics[width=90mm, height=66mm]{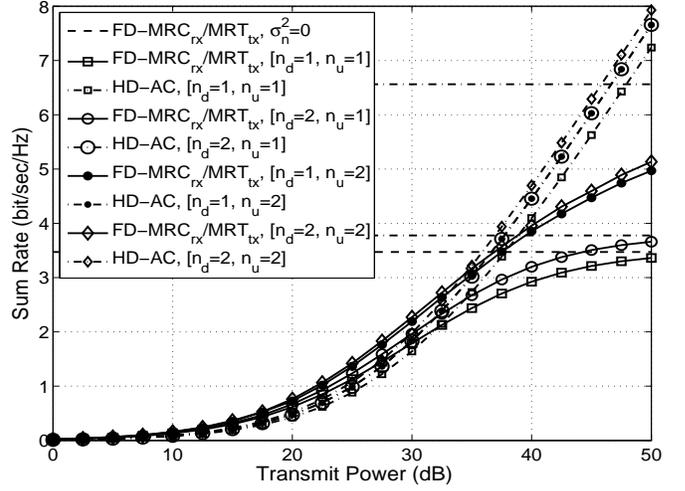}
\vspace{-1em} \caption{ UL and DL average sum rate versus transmit power for the FD and HD AP with different antenna configurations ($d=25$ m, $P_{\AP}=2P_{\Us}$, and $\sigma_{\AP\AP}^2=0.1$).}
\vspace{-1em}
\label{fig:sum_rate_Antenna Config}
\end{figure}
\appendices
\section{Proof of Lemma~\ref{propos: cdf for SINRd uplink MRC/MRT case1 alpha2}}
\label{APX: propos: cdf for SINRd uplink MRC/MRT case1 alpha2 }
Following~\eqref{eq: cdf of SINRd polar coordinates general}, the $F_{\SINRd}(z)$ corresponding to $\alpha=2$ and $\Sn=0$ is given by
\begin{align}
F_{\SINRAP}(z)
&=1-\frac{P_{\Us}}{P_{\AP}} \frac{ 1}{z\Sap}\int_{0}^{R_c}\\
&\quad\times
\int_{0}^{2\pi}\frac{\lambda_{\Ds}r e^{-\lambda_{\Ds}\pi r^2}}
{\frac{P_{\Us}}{P_{\AP}} \frac{ 1}{z\Sap} + r^2+d^2-2rd\cos\theta}d\theta dr.\nonumber
\end{align}
With the help of~\cite[Eq. (3.661.4)]{Integral:Series:Ryzhik:1992}, and making the change of variable $r^2=\upsilon$, we obtain
\vspace{-0.2em}
\begin{align}\label{eq: cdf of SINRd integral over r}
&F_{\SINRAP}(z)
=\!\!1-\frac{P_{\Us}}{P_{\AP }} \frac{ \pi\lambda_{\Ds} }{z\Sap}
\int_{0}^{R_c^2}\!\!
\frac{ e^{-\lambda_{\Ds}\pi \upsilon}}
{\sqrt{ \upsilon^2 + 2 b\upsilon + c}}d\upsilon.
\end{align}
To the best of our knowledge, the integral in~\eqref{eq: cdf of SINRd integral over r} does not admit a
closed-form solution. In order to proceed, we use Taylor series representation~\cite[Eq. (1.211.1)]{Integral:Series:Ryzhik:1992} for the term $e^{-\lambda_{\Ds}\pi \upsilon}$, and write
\vspace{-0.5em}
\begin{align}\label{eq: cdf of SINRd integral over r: Taylor series}
F_{\SINRAP}(z)
&=1-\frac{P_{\Us}}{P_{\AP }} \frac{ \pi\lambda_{\Ds} }{z\Sap}
\sum_{k=0}^{\infty}\frac{(-\lambda_{\Ds}\pi)^k}{k!}
\!\!\int_{0}^{R_c^2}\!
\frac{  \upsilon^k}
{\sqrt{ \upsilon^2 + 2 b\upsilon + c}}d\upsilon.
\end{align}
A change of variable $\sqrt{\upsilon^2 \!+ \!2 b\upsilon\! + c}=\!\upsilon t + \!\sqrt{c}$, and after some manipulations,~\eqref{eq: cdf of SINRd integral over r: Taylor series} can be expressed as
\vspace{-0.2em}
\begin{align}
&F_{\SINRAP}(z)
=1-\frac{P_{\Us}}{P_{\AP }} \frac{ 4\pi\lambda_{\Ds} }{z\Sap}
\sum_{k=0}^{\infty}\frac{(-\lambda_{\Ds}\pi)^k}{k!}
\!\!\int_{\frac{b}{\sqrt{c}}}^{\varrho}\!
\frac{(b-\sqrt{c}t)^{k}}{ (t^2-1)^{k+1}}
dt.\nonumber
\end{align}
Finally, using~\cite[Eq. (5.8.2)]{Transcendental:book}, we arrive at the desired result given in~\eqref{eq: cdf of SINRd integral over r: alpha 2 Final}.
\vspace{-0.7em}
\section{Proof of Lemma~\ref{lemma:cdf for SINRd uplink MRC/MRT case1 alpha4}}
\label{APX: propos: cdf for SINRd uplink MRC/MRT case1 alpha4 }
Following~\eqref{eq: cdf of SINRd polar coordinates general}, the $F_{\SINRAP}(z)$ corresponding to $\alpha=4$ and $\Sn=0$ can be written as
\vspace{-0.6em}
\begin{align}\label{eq: proof of the cdf of SINRd alpha 4 over r and theta}
F_{\SINRAP}(z)
 &=1-\frac{ 1}{z}
 \frac{ \lambda_{\Ds}P_{\Us}}{\Sap P_{\AP}}\\
 &~\times
\int_{0}^{R_c}
\!\int_{0}^{2\pi}\!\frac{re^{-\lambda_{\Ds}\pi r^2}}
{\frac{P_{\Us}}{P_{\AP}} \frac{ 1}{z\Sap}\! +\!(r^2+d^2-2rd\cos\theta)^{2}}d\theta dr\!.\nonumber
\end{align}
By using~\cite{Integral:Series:Ryzhik:1992}, the inner integral can be obtained as
\vspace{-0.6em}
\begin{align}\label{eq: proof of the cdf of SINRd alpha 4 over r}
&F_{\SINRAP}(z)
=1-\frac{ \sqrt{2}\pi}{z}
 \frac{ \lambda_{\Ds}P_{\Us}}{\Sap P_{\AP}}\\
 &~\times\int_{0}^{R_c}\!\!\!
\frac{re^{-\lambda_{\Ds}\pi r^2}}
{\sqrt{c_2(r) \!\!+ \!\!\sqrt{c_4(r)c_0(r)}}}
\left(\frac{1}{\sqrt{c_0(r)}}\!\! + \!\!\frac{1}{\sqrt{c_4(r)}}\right) dr\!,\nonumber
\end{align}
where $c_0(r)=b_0(r)-b_1(r)+b_2(r)$, $c_2(r)=b_0(r)-b_2(r)$, and $c_4(r)=b_0(r)+b_1(r)+b_2(r)$, with $b_0(r)\!=P_{\Us}/( P_{\AP}z\Sap)\! +\! (r^2\!+d^2)^2$, $b_1(r)\!=\!4rd(r^2\!+d^2)$, and $b_2(r)\!=\!4r^2d^2$. We can simplify the above integral in the case of $d=0$. Hence, after a simple substitution $r^2=\upsilon$,~\eqref{eq: proof of the cdf of SINRd alpha 4 over r} can be written as
\vspace{-0.5em}
\begin{align}\label{eq: proof of the cdf of SINRd alpha 4 over r simplified}
&F_{\SINRAP}(z)
> 1-\frac{ \pi}{z}
 \frac{ \lambda_{\Ds}P_{\Us}}{\Sap P_{\AP}}
\int_{0}^{R_c^2}
\frac{e^{-\lambda_{\Ds}\pi \upsilon}}
{\upsilon^2+\frac{ P_{\Us}}{ P_{\AP}} \frac{1}{z\Sap}}d\upsilon.
\end{align}
In order to simplify \eqref{eq: proof of the cdf of SINRd alpha 4 over r simplified}, we adopt a series expansion of the
exponential term. Substituting the series expansion of $e^{-\lambda_{\Ds}\pi \upsilon}$ into the~\eqref{eq: proof of the cdf of SINRd alpha 4 over r simplified} yields
\vspace{-0.6em}
\begin{align}
F_{\SINRAP}(z)
&\!> \!1\!-\!\frac{ P_{\Us}}{ P_{\AP}} \frac{1}{z\Sap} \!\sum_{k=0}^{\infty}\!\frac{ (-1)^k(\lambda_{\Ds}\pi)^{k+1}}{k!}
\!\!\int_{0}^{R_c^2}\!\!\!\!
\frac{ \upsilon^k}
{\upsilon^2\!+\!\frac{ P_{\Us}}{ P_{\AP}} \frac{1}{z\Sap}}d\upsilon.\nonumber
\end{align}
Let us denote $\beta=\frac{P_{\Us}}{ P_{\AP}} \frac{1}{z\Sap}$.
By making the change of variable $\left(\upsilon/R_c^2\right)^2=t$, we obtain
\vspace{-0.6em}
\begin{align}\label{eq: proof of the cdf of SINRd alpha 4 over r taylor}
F_{\SINRAP}(z)
&\!> \!1\!-\!
 \sum_{k=0}^{\infty}\frac{(-1)^k(\lambda_{\Ds}\pi R_c^2)^{k+1}}{2k!}
\!\!\int_{0}^{1}\!\!
\frac{ t^{\frac{k-1}{2}}}
{1\!+\frac{R_c^4}{\beta}t}dt.
\end{align}
Now with the help of~\cite[Eq. (9.111)]{Integral:Series:Ryzhik:1992} the integral in \eqref{eq: proof of the cdf of SINRd alpha 4 over r taylor} can be solved to yield \eqref{eq: proof of the cdf of SINRd alpha final}.
\vspace{-0.4em}
\balance

\bibliographystyle{IEEEtran}

\begin{thebibliography}{1}

\bibitem{DBLP:journals/jsac/SabharwalSGBRW14a}
A. Sabharwal, P. Schniter, D. Guo, D. W. Bliss, S. Rangarajan, and R. Wichman, ``In-band full-duplex wireless: Challenges and opportunities,'' \emph{
IEEE J. Sel. Areas Commun.}, vol. 32, pp. 1637-1652, Sep. 2014.

\bibitem{Riihonen:WCOM:2009}
T. Riihonen, S. Werner, and R. Wichman, ``Optimized gain control for
single-frequency relaying with loop interference," \emph{IEEE Trans. Wireless
Commun.,} vol. 8, pp. 2801-2806, June 2009.

\bibitem{Riihonen:WCOm:2011}
-----``Hybrid full-duplex/half-duplex relaying with transmit power
adaptation,"\emph{ IEEE Trans. Wireless Commun.}, vol. 10, pp. 3074-3085,
Sep. 2011.

\bibitem{Riihonen:JSP:2011}
-----, ``Mitigation of loopback self-interference in full-duplex MIMO relays,'' \emph{ IEEE Trans. Signal Process.}, vol. 59, pp. 5983-5993, Dec. 2011.


\bibitem{Duarte:PhD:dis}
M. Duarte, ``Full-duplex wireless: Design, implementation and characterization,'' Ph.D. dissertation, Dept. Elect. and Computer Eng., Rice
University, Houston, TX, 2012.

\bibitem{Bliss:Asilomar:2012}
D. W. Bliss, T. Hancock, and P. Schniter, ``Hardware and environmental
phenomenological limits on full-duplex MIMO relay performance,"
in \emph{Proc.46th Asilomar Conf. on Signals, Systems, and Computers
(ASILOMAR 2012)}, Paciﬁc Grove, CA, Nov. 2012, pp. 34-39.


\bibitem{Khojastepour:Mobicom:2012}
E. Aryafar, M. A. Khojastepour, K. Sundaresan, S. Rangarajan, and
M. Chiang, ``MIDU: Enabling MIMO full duplex,''in \emph{ Proc. 18th Intl.
Conf. Mobile Computing and Networking (ACM Mobicom ’12)}, New York, NY, Aug. 2012, pp. 257-268.


\bibitem{Sachin:NSDI:2014}
D. Bharadia and S. Katti, ``Full-duplex MIMO radios,'' in\emph{ Proc. 11th
USENIX Symp. Networked Syst}. Design and Implementation (NSDI
’14), Seattle, WA, Apr. 2014, pp. 359-372.


\bibitem{Everett:JWCOM:2014}
E. Everett, A. Sahai, and A. Sabharwal, ``Passive self-interference
suppression for full-duplex infrastructure nodes," \emph{IEEE Trans. Wireless
Commun.}, vol. 13, pp. 680-694, Feb. 2014.


\bibitem{Korpi:JSAC:2014}
 D. Korpi, L. Anttila, V. Syrj\"{a}l\"{a}, and M. Valkama, ``Widely linear
digital self-interference cancellation in direct-conversion full-duplex
transceiver," \emph{IEEE J. Sel. Areas Commun.}, vol. 32, pp. 1674-1687, Sep.
2014.



\bibitem{Himal:WCOM:FD:2014}
H. A. Suraweera, I. Krikidis, G. Zheng, C. Yuen, and P. J. Smith, ``Lowcomplexity
end-to-end performance optimization in MIMO full-duplex
relay systems," \emph{IEEE Trans. Wireless Commun.}, vol. 13, pp. 913-927,
Jan. 2014.

\bibitem{Leonardo:JSAC:2014}
L. J. Rodr\'{\i}guez, N. H. Tran, and T. Le-Ngoc, ``Performance of fullduplex
AF relaying in the presence of residual self-interference," \emph{IEEE
J. Sel. Areas Commun.}, vol. 32, pp. 1752-1764, Sept. 2014.


\bibitem{Ngo:JSAC:2014}
H. Q. Ngo, H. A. Suraweera, M. Matthaiou, and E. G. Larsson, ``Multipair
full-duplex relaying with massive arrays and linear processing,'' \emph{
IEEE J. Sel. Areas Commun.}, vol. 32, pp. 1721-1737, Sep 2014.


\bibitem{ZhenJSP:2015}
G. Zheng, ``Joint beamforming optimization and power control for fullduplex
MIMO two-way relay channel," \emph{IEEE Trans. Signal Process.,}
vol. 63, pp. 555-566, Feb. 2015.

\bibitem{Riihonen:CROWNCOM:2014}
T. Riihonen and R. Wichman, ``Energy detection in full-duplex cognitive
radios under residual self-interference," in \emph{Proc. 9th Intl. Conf. Cognitive
Radio Oriented Wireless Networks (CROWNCOM),} Oulu, Finland, June
2014, pp. 57-60.

\bibitem{Zheng:JSPL2013}
G. Zheng, I. Krikidis, J. Li, A. P. Petropulu, and B. E. Ottersten,
``Improving physical layer secrecy using full-duplex jamming receivers,"
\emph{IEEE Trans. Signal Process.,} vol. 61, pp. 4962-4974, Oct. 2013.


\bibitem{Sanjay:CISS:2013}
S. Goyal, P. Liu, S. Hua, and S. S. Panwar, ``Analyzing a full-duplex
cellular system,'' in\emph{ Proc. 47th Annual Conf. on Information Sciences
and Systems (CISS)}, Baltimore, MD, Mar. 2013, pp. 1–6.


\bibitem{Girnyk:2013}
M. Vehkaper a, M. Girnyk, T. Riihonen, R. Wichman, and L. Rasmussen,
``On achievable rate regions at large-system limit in full-duplex
wireless local access,'' in\emph{ Proc. 1st Intl. Black Sea Conf. Commun. and
Networking (BlackSeaCom)}, Batumi, Georgia, July 2013, pp. 7-11.


\bibitem{Sundaresan:Mobicom:2014}
K. Sundaresan, M. Khojastepour, E. Chai, and S. Rangarajan, ``Full-duplex
without strings: Enabling full-duplex with half-duplex clients,"
in \emph{Proc. 20th Intl. Conf. Mobile computing and networking (MobiCom
2014)}, Maui, HI, Sep. 2014, pp. 55-66.


\bibitem{Panwar:ICC14}
S. Goyal, P. Liu, S. S. Panwar, R. A. DiFazio, R. Yang, J. Li, and
E. Bala, ``Improving small cell capacity with common-carrier full
duplex radios,'' in\emph{ Proc. IEEE Intl. Conf. Commun. (ICC 2014)}, Sydney,
Australia, June 2014, pp. 4987-4993.

\bibitem{DBeiYin:ACSSC}
B. Yin, M. Wu, C. Studer, J. R. Cavallaro, and J. Lilleberg, ``Full-duplex
in large-scale wireless systems,'' in\emph{ Proc. Asilomar Conf.
Signals, Systems and Computers (ASILOMAR 2013)}, Paciﬁc Grove,
CA, Nov. 2013, pp. 1623-1627.

\bibitem{Achaleshwar:acssc:DS13}
A. Sahai, S. Diggavi, and A. Sabharwal, ``On uplink/downlink fullduplex
networks,'' in\emph{ Proc. Asilomar Conf. Signals, Systems and
Computers (ASILOMAR 2013)}, Paciﬁc Grove, CA, Nov. 2013, pp. 14-18.


\bibitem{Full:Nguyen:JSP:2013}
D. Nguyen, L.-N. Tran, P. Pirinen, and M. Latva-aho, ``Precoding for
full duplex multiuser MIMO systems: Spectral and energy efficiency
maximization,'' \emph{ IEEE Trans. Signal Process.}, vol. 61, pp. 4038-4050,
Aug. 2013.


\bibitem{JeminLee:2014}
J. Lee and T. Q. S. Quek, ``Hybrid full-/half-duplex system analysis
in heterogeneous wireless networks," \emph{IEEE Trans. Wireless Commun.,}
vol. 14, pp. 2883-2895, May 2015.

\bibitem{Haenggi:arXiv}
Z. Tong and M. Haenggi, ``Throughput analysis for full-duplex wireless
networks with imperfect self-interference cancellation," \emph{IEEE Trans.
Commun.,} 2015.

\bibitem{Venkateswaran:Infocom:2015}
S. Wang, V. Venkateswaran, and X. Zhang, ``Exploring full-duplex gains
in multi-cell wireless networks: A spatial stochastic framework," \emph{in Proc.
IEEE Conference on Computer Communications (INFOCOM 2015)}, Hong Kong, Apr. 2015, pp. 855-863.

\bibitem{Haenggi:TWC:ICIC:2014}
X. Zhang and M. Haenggi, ``A stochastic geometry analysis of inter-cell
interference coordination and intra-cell diversity," \emph{IEEE Trans. Wireless
Commun.}, vol. 13, pp. 6655-6669, Dec. 2014.


\bibitem{Integral:Series:Ryzhik:1992}
I. S. Gradshteyn and I. M. Ryzhik, \emph{Table of Integrals, Series and Products}. $7$th ed. Academic Press, 2007.

\bibitem{Transcendental:book}
A. Erdelyi, \emph{ Higher Transcendental Functions. New York: McGraw-
Hill}, 1953, vol. 1.


 \bibitem{Duplo:Techrep}
``European FP7 project DUPLO (Full-Duplex Radios for Local Access),"
European Commission - Research: The Seventh Framework Programme,
http://www.fp7-duplo.eu/index.php/general-info, Tech. Rep., Nov. 2012.


\bibitem{Andrews:MCOM:2013}
J. G. Andrews, ``Seven ways that HetNets are a cellular paradigm shift,'' \emph{
IEEE Commun. Mag.}, vol. 51, pp. 136-144, Mar. 2013.


 \bibitem{Haenggi:IT:2005}
M. Haenggi, ``On distances in uniformly random networks,'' \emph{ IEEE
Trans. Inf. Theory}, vol. 51, pp. 3584-3586, Oct. 2005.



 \bibitem{TitusMRT:JCOM:1999}
T. K. Y. Lo, ``Maximum ratio transmission," IEEE Trans. Commun.,
vol. 47, pp. 1458-1461, Oct. 1999.


\bibitem{Adamchik:1990}
V. S. Adamchik and O. I. Marichev, ``The algorithm for calculating
integrals of hypergeometric type functions and its realization in RE-
DUCE system,'' in\emph{ in Proc. Int. Conf. Symbolic and Algebraic Comput.},
Tokyo, Japan, 1990, pp. 212-224.


 \bibitem{Prudnikov:vol3}
A. P. Prudnikov, Y. A. Brychkov, and O. I. Marichev, \emph{Integral and
Series. Vol. 3: More Special Functions}. Amsterdam: Gordon and Breach
Science Publishers, 1986.

 \bibitem{MathematicalStatistics_1978}
R. V. Hogg and A. T. Craig, \emph{Introduction to Mathematical Statistics,}
4th ed. Macmillan, New York, 1978.

 \bibitem{Xiaohu:Eigen:IT:2007}
S. Jin, X. Gao, and X. You, ``On the ergodic capacity of rank-1 Ricean fading
MIMO channels," \emph{IEEE Trans. Inf. Theory,} vol. 53, pp. 502-517,
Feb. 2007.

 \bibitem{Johnson:Book:1990}
R. A. Horn and C. A. Johnson, \emph{Matrix Analysis}. Cambridge, U.K.:
Cambridge Univ. Press,  1990.

 \bibitem{Zhi-quan:SDP:2010}
Z.-Q. Luo, W.-K. Ma, A.-C. So, Y. Ye, and S. Zhang, ``Semidefinite
relaxation of quadratic optimization problems," \emph{IEEE Signal Process.
Mag.}, vol. 27, pp. 20-34, May 2010.

 \bibitem{Quoc:2013:COM}
H. Q. Ngo, E. G. Larsson, and T. L. Marzetta, ``Energy and spectral efficiency
of very large multiuser MIMO systems," \emph{IEEE Trans. Commun.,}
vol. 61, pp. 1436-1449, Apr. 2013.

 \bibitem{Paulraj:2007}
M. Charafeddine, A. Sezgin, and A. Paulraj,``Rate region frontiers for
n-user interference channel with interference as noise," in \emph{Proc. 45th
Allerton Conf. Communication, Control, and Computing}, Monticello,
IL, Sep. 2007.








\end{thebibliography}

\end{document}